\documentclass[sigconf, nonacm, pdfa]{acmart}
\usepackage[mathscr]{eucal}
\usepackage{pifont}
\newcommand{\cmark}{\ding{51}}%
\newcommand{\xmark}{\ding{55}}%
\usepackage{ifthen}
\usepackage{multirow}
\usepackage{enumitem}

\usepackage{algorithm}
\usepackage{algorithmic}

\newcommand\ALG[3]{\begin{algorithm}[t]\caption{#2}\label{alg:#1}\begin{algorithmic}[1]#3\end{algorithmic}\end{algorithm}}
\newcommand\ALGref[1]{Algorithm~\ref{alg:#1}}

\usepackage{subcaption}
\usepackage{balance}

\newcommand\Pa[1]{{(#1)}}
\newcommand\VEC{\operatorname{vec}}
\newcommand\Mat[1]{\mathop{\operatorname{mat}_{#1}}}

\newcommand\TODO[1][]{{\color{orange}[TODO\ifthenelse{\equal{#1}{}}{}{: #1}]}}
\newcommand\SEC\section
\newcommand\SSEC\subsection
\newcommand\SSSEC\subsubsection
\newcommand\Par[1]{\vspace{0.5em}\noindent\textbf{#1.}}

\newcommand\BM\boldsymbol
\newcommand\BB\mathbb
\newcommand\CAL\mathcal
\newcommand\SCR\mathscr
\newcommand\FRAK\mathfrak
\newcommand\BS[1]{\boldsymbol{\mathscr{#1}}}

\newcommand\TP{\mathsf{T}}
\newcommand\Tp{\TP}

\newcommand\MAT[1]{\left[\begin{smallmatrix}#1\end{smallmatrix}\right]}

\newcommand\AL[1]{\begin{align}#1\end{align}}

\newcommand\HAT\widehat
\newcommand\BAR\overline
\newcommand\TLD\widetilde

\usepackage[a-2b]{pdfx}
\usepackage{balance}
\newcommand\Ours{TUCKET}
\newcommand\OursFull{\underline{Tucke}r \underline{T}ree}
\newcommand\FullPaperUrl{\url{https://github.com/q-rz/TUCKET/blob/main/TUCKET-Full.pdf}}
\newcommand\ZoomTucker{Zoom-Tucker}
\newcommand\Ts{T_\textnormal{s}}
\newcommand\Te{T_\textnormal{e}}

\newcommand{\change}[1]{{\textcolor{blue}{#1}}}

\newcommand{\ie}{{i.e.}}

\renewcommand\change[1]{#1}

\newcommand\vldbdoi{10.14778/3704965.3704980}
\newcommand\vldbpages{4746 - 4759}
\newcommand\vldbvolume{17}
\newcommand\vldbissue{13}
\newcommand\vldbyear{2024}
\newcommand\vldbauthors{\authors}
\newcommand\vldbtitle{\shorttitle} 
\newcommand\vldbavailabilityurl{https://github.com/q-rz/TUCKET}
\newcommand\vldbpagestyle{empty} 

\begin{document}

\title{\Ours: A Tensor Time Series Data Structure for Efficient and Accurate Factor Analysis over Time Ranges}

\author{Ruizhong Qiu}
\authornote{Equal contribution.}
\affiliation{%
  \institution{University of Illinois}
  \city{Urbana--Champaign}
}
\email{rq5@illinois.edu}

\author{Jun-Gi Jang}
\authornotemark[1]
\affiliation{%
  \institution{University of Illinois}
  \city{Urbana--Champaign}
}
\email{jungi@illinois.edu}

\author{Xiao Lin}
\affiliation{%
  \institution{University of Illinois}
  \city{Urbana--Champaign}
}
\email{xiaol13@illinois.edu}

\author{Lihui Liu}
\affiliation{%
  \institution{University of Illinois}
  \city{Urbana--Champaign}
}
\email{lihuil2@illinois.edu}

\author{Hanghang Tong}
\affiliation{%
  \institution{University of Illinois}
  \city{Urbana--Champaign}
}
\email{htong@illinois.edu}

\renewcommand{\shortauthors}{Qiu et al.}
\setlength{\floatsep}{0.01cm}
\setlength{\textfloatsep}{0.01cm}
\setlength{\intextsep}{0.01cm}
\setlength{\dblfloatsep}{0.01cm}
\setlength{\dbltextfloatsep}{0.01cm}
\setlength{\abovedisplayskip}{0.01cm}
\setlength{\belowdisplayskip}{0.01cm}
\setlength{\abovecaptionskip}{0.01cm}
\setlength{\belowcaptionskip}{0.01cm}

\begin{abstract}

Given an evolving tensor time series and multiple time ranges, how can we compute Tucker decomposition for each time range efficiently and accurately? 
Tucker decomposition has been widely used in a variety of applications to obtain latent factors of tensor data. 
For example, Tucker decomposition on air pollution data allows us to analyze and compare air pollution patterns between different locations during different periods of time.
In these applications, a common need is to compute Tucker decomposition for a given time range. Furthermore, real-world tensor time series are typically evolving in the time dimension. 
Such needs call for a data structure that can efficiently and accurately support range queries of Tucker decomposition and stream updates.
Unfortunately, existing 
methods do not support either range queries 
or stream updates. 
For methods that do not support range queries, they have to re-compute from scratch for each query. Not until 2021 has a data structure called Zoom-Tucker been proposed to support range queries via block-wise preprocessing. However, Zoom-Tucker does not support stream updates and, more critically, suffers from a reluctant efficiency--accuracy tradeoff 
--- a large block size causes inaccuracy, while a small block size leads to inefficiency. This challenging problem has remained open for years prior to our work.
To solve this challenging problem, we propose \Ours{}, a data structure that can efficiently and accurately handle both range queries 
and stream updates. 
Our key idea is to design a new data structure that we call a \emph{stream segment tree} by generalizing the \emph{segment tree}, 
a data structure that was originally invented for computational geometry.
For a range query of length $L$, our \Ours{} can find $O(\log L)$ nodes (called the \emph{hit set}) from the tree and efficiently stitch their preprocessed decompositions to answer the range query. We also propose an algorithm to optimally prune the hit set via an approximation of 
subtensor 
decomposition. For the $T$-th stream update, our \Ours{} modifies only amortized $O(1)$ nodes and only $O(\log T)$ nodes in the worst case. 
Extensive evaluation demonstrates that our \Ours{} consistently achieves the highest efficiency and accuracy across four large-scale datasets. Our \Ours{} achieves \change{at least 3 times lower latency and at least 1.4 times smaller reconstruction error than Zoom-Tucker on all datasets}. 
The full version can be found at \FullPaperUrl.
\end{abstract}

\maketitle

\pagestyle{\vldbpagestyle}
\begingroup\small\noindent\raggedright\textbf{PVLDB Reference Format:}\\
\vldbauthors. \vldbtitle. PVLDB, \vldbvolume(\vldbissue): \vldbpages, \vldbyear.\\
\href{https://doi.org/\vldbdoi}{doi:\vldbdoi}
\endgroup
\begingroup
\renewcommand\thefootnote{}\footnote{\noindent
This work is licensed under the Creative Commons BY-NC-ND 4.0 International License. Visit \url{https://creativecommons.org/licenses/by-nc-nd/4.0/} to view a copy of this license. For any use beyond those covered by this license, obtain permission by emailing \href{mailto:info@vldb.org}{info@vldb.org}. Copyright is held by the owner/author(s). Publication rights licensed to the VLDB Endowment. \\
\raggedright Proceedings of the VLDB Endowment, Vol. \vldbvolume, No. \vldbissue\ %
ISSN 2150-8097. \\
\href{https://doi.org/\vldbdoi}{doi:\vldbdoi} \\
}\addtocounter{footnote}{-1}\endgroup

\ifdefempty{\vldbavailabilityurl}{}{
\vspace{.3cm}
\begingroup\small\noindent\raggedright\textbf{PVLDB Artifact Availability:}\\
The source code, data, and/or other artifacts have been made available at \url{\vldbavailabilityurl}.
\endgroup
}

\begin{figure}[t]
\centering
\captionsetup[subfigure]{justification=centering}
\setcounter{subfigure}{0}
\subfloat[March 2015]{\includegraphics[width=0.319\linewidth]{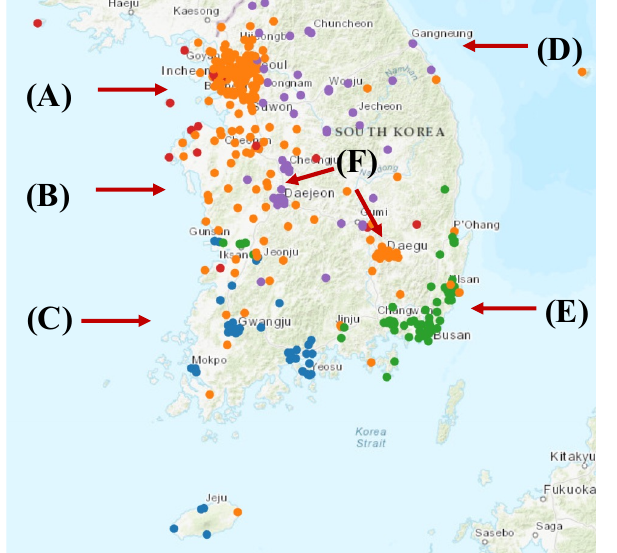}\label{fig:exp-case-2015}}
\subfloat[March 2016]{\includegraphics[width=0.32\linewidth]{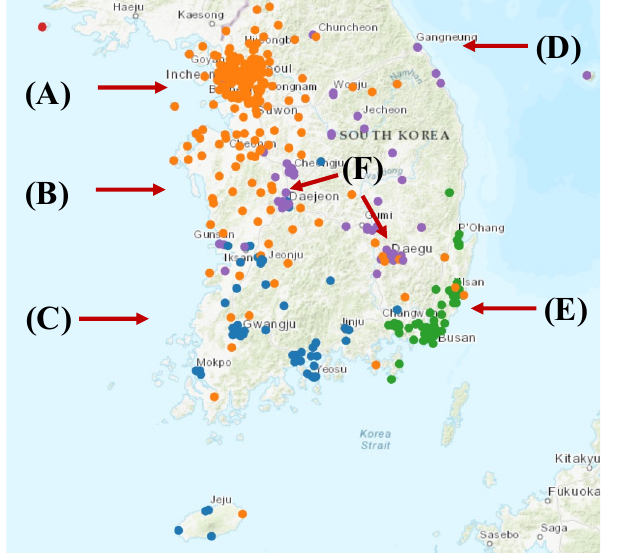}\label{fig:exp-case-2016}}
\subfloat[March 2017]{\includegraphics[width=0.324\linewidth]{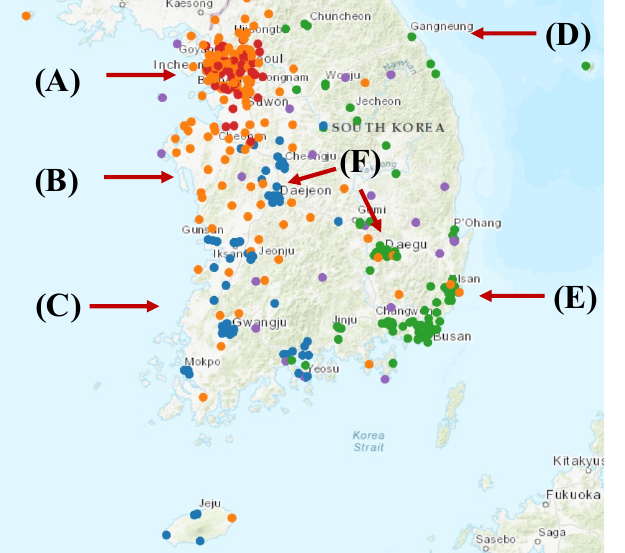}\label{fig:exp-case-2017}}
\caption{\change{
Case study on Air Quality data (see Section~\ref{subsec:exp-case-study})
}}
\label{fig:exp-case-air}
\end{figure}

\section{Introduction}

Tensor time series are ubiquitous in the real world, ranging from multimedia data such as videos and music to time series data such as stock prices, traffic volumes, climate, agriculture, environmental monitoring, and physical systems.
Tensor decomposition, such as CANDECOMP \cite{carroll1970analysis} / PARAFAC \cite{harshman1970foundations} (CP), PARAFAC2 \cite{kiers1999parafac2}, and Tucker \cite{tuckerals} decompositions, is a fundamental approach to tensor data analysis and performs an essential role in various applications including clustering \cite{zhao2005tricluster,huang2008simultaneous,cao2014robust}, dimension reduction \cite{wang2008tensor,kim2015compression}, anomaly detection \cite{koutra2012tensorsplat,fanaee2016tensor}, concept discovery \cite{jeon2015haten2,jeon2016scout,ahn2020gtensor}, and so on~\cite{yao2015context,spelta2017financial,de2017tensorcast,lacroix2019tensor}. 
As a generalization of singular value decomposition, Tucker decomposition \cite{tuckerals} seeks to approximately factorize a tensor into \emph{factor matrices} for each mode of the tensor and a small \emph{core tensor} characterizing the relations of the factor matrices. 
\change{Factor matrices and the core tensor can serve as the input for downstream data mining algorithms such as clustering \cite{zhao2005tricluster} and anomaly detection \cite{koutra2012tensorsplat}. 
}


\begin{table*}[t]
\centering
\caption{Comparison in functionalities, time complexities per query, and overall space complexities. See Table~\ref{tab:nomen} for the definitions of $L,D,T,r,p$. 
Empirically, our \Ours{} is faster than all other methods on a GPU because it is more parallelizable. Zoom-Tucker cannot use a large block size $b$, or otherwise it would incur high error. 
}
\label{tab:functionality}
\begin{tabular}{l|cc|ll}
\toprule
\textbf{Method}&\textbf{Range Query}&\textbf{Stream Update}&\textbf{Time Complexity per Query}&\textbf{Space Complexity}\\
\midrule
Tucker-ALS \cite{tuckerals}&\xmark&\xmark& $O(r\textcolor{red}{D^{p-1}L})$ &$O(\textcolor{red}{D^{p-1}}T)$\\
D-Tucker~\cite{dtucker}&\xmark&\cmark&$O(r^2\textcolor{red}{D^{p-2}L})$ &$O((\textcolor{red}{D^{p-1}}+r\textcolor{red}{D^{p-2}})T)$\\
Zoom-Tucker~\cite{zoomtucker}&\cmark&\xmark&$O\big(r^2\textcolor{red}{D\frac Lb}+r^2L+r^{p+1}\frac Lb\big)$&$O\big((r\textcolor{blue}{D}+r^p)\frac Tb+rT\big)$\\
\midrule
\textbf{\Ours{}} (ours)&\cmark&\cmark&$O(r^p\textcolor{blue}{D\log L}+r^{2p-2}(D+L)+\log T)$&$O((r\textcolor{blue}{D}+r^p)T+rT\log T)$\\
\bottomrule
\end{tabular}
\end{table*}

In the analysis of tensor time series, a common situation is to discover latent patterns in given time ranges \cite{jang2018zoom}. 
\change{
For example, given air quality data (a 3-way tensor time series $\BS X$ where $\SCR X_{t,i,j}$ represents the concentration value of air pollutant $j$ in location $i$ at time $t$), environmental scientists can find out which locations share similarity pollution patterns in March of each year by analyzing the Tucker decomposition of each month. (See Figure~\ref{fig:exp-case-air} for illustration and the case study in Section~\ref{subsec:exp-case-study} for detail.) 
Range queries are necessary here because Tucker decompositions vary across different time ranges due to the evolving nature of tensor time series, as shown in Figure~\ref{fig:exp-case-air}.}
Such needs give rise to an interesting research question: given an evolving tensor time series and multiple time range queries, how can we design a data structure that can efficiently and accurately compute Tucker decomposition for each time range?

Unfortunately, existing Tucker decomposition methods do not support either range queries or stream updates (see Table~\ref{tab:functionality}). For methods that do not support range queries, they have to re-compute from scratch for each range query. For example, D-Tucker \cite{dtucker} handles stream updates via slice-wise preprocessing, but it supports only full Tucker decomposition and cannot answer range queries efficiently. 
Not until 2021 has a method called Zoom-Tucker \cite{zoomtucker} been proposed to support range queries. 
Zoom-Tucker consists of two phases: a preprocessing phase and a query phase. First, the preprocessing phase divides the timespan into blocks and preprocesses the Tucker decomposition of each block. Next, the query phase answers time range queries by 
stitching the preprocessed blocks included in the query range. However, Zoom-Tucker does not support efficient stream updates (i.e., appending a new tensor slice) due to its block structure. 

Moreover, Zoom-Tucker suffers from a critical limitation: 
a large block size causes low accuracy for short ranges due to a high approximation error, while a small block size leads to inefficiency for long ranges that require to stitch many blocks.
It means that \ZoomTucker{} suffers from a reluctant tradeoff between accuracy and efficiency. How to avoid this tradeoff has been an open problem for years. Prior to our work, no existing method achieves both efficiency and accuracy for Tucker decomposition range queries. A crucial challenge here is how to design a more sophisticated data structure and organize 
preprocessed results to avoid the efficiency--accuracy tradeoff. What makes it even more challenging is that the data structure needs to efficiently support stream updates to the tensor time series. 

To solve this challenging problem, we propose a new data structure 
called \emph{\OursFull} (\Ours) that can efficiently and accurately handle both range queries of Tucker decomposition and stream updates. The key idea of our \Ours{} is to design a new data structure that we call a \emph{stream segment tree} by generalizing the \emph{segment tree} \cite{bentley1977algorithms}, a data structure that was originally invented for computational geometry.
For a range query of length $L$, our \Ours{} can find $O(\log L)$ nodes (called the \emph{hit set}) from the tree via our optimal pruning algorithm and efficiently stitch their preprocessed decompositions to answer the range query. 
Besides that, for the $T$-th stream update, our \Ours{} modifies only amortized $O(1)$ nodes and only $O(\log T)$ nodes in the worst case. 

The main contributions of this paper are summarized as follows: 
\begin{itemize}[noitemsep,topsep=0pt]
\item\change{\textbf{Data structure.} We design a new data structure \emph{stream segment tree} to efficiently handle stream updates. It is much faster here than the interval tree \cite{preparata2012computational} and the R-tree \cite{guttman1984r}. 
}
\item\change{\textbf{Hit set pruning algorithm.} We derive an efficient scheme to approximate subtensor decompositions and employ it to further reduce the size of the hit set for each query compared with the standard segment tree.}
\item\change{\textbf{Stitching algorithm.} We propose a new algorithm for stitching subtensor decompositions. Our stitching algorithm is more GPU-parallelizable and more numerically stable than Zoom-Tucker's stitching algorithm.}
\item\textbf{Theoretical guarantees.} We provide detailed theoretical guarantees for our proposed method in terms of time complexity, space complexity, and error analysis. 
\item\textbf{Empirical evaluation.} We conduct extensive experiments to evaluate our \Ours{} against state-of-the-art methods for Tucker decomposition on large-scale real-world tensor time series datasets. Our \Ours{} consistently achieves both the highest efficiency and the highest accuracy across all datasets. For example, our \Ours{} achieves at least 3 times lower latency and at least 1.4 times smaller reconstruction error on all datasets. 
\end{itemize}

\section{Preliminaries}
In this section, we present the preliminaries 
on Tucker decomposition and Tucker-ALS. 
Main symbols used in this paper are summarized in Table~\ref{tab:nomen}. Due to the space limit, preliminaries on basic tensor operations are deferred to the full version. 

\begin{table}[t]
\begin{center}\begin{small}
\caption{Nomenclature.}
\label{tab:nomen}
\begin{tabular}{ll}
\toprule
\textbf{Symbol}&\textbf{Description}\\
\midrule
$\VEC$&vectorization\\
$\Mat n$&mode-$n$ matricization\\
$\times_n$&mode-$n$ tensor-matrix product\\
$\|\cdot\|_\text F$&Frobenius norm\\
$\otimes$&matrix Kronecker product\\
$\,{}^\Tp$&matrix transpose\\
\midrule
$\BS X$ & a tensor time series\\
$p$ & number of modes of $\BS X$\\
$T$ & size of the temporal mode of $\BS X$\\
$D_2,\dots,D_p$ & sizes of non-temporal modes of $\BS X$\\
$D:=\max\{D_2,\dots,D_p\}$& maximum size of non-temporal modes\\
\midrule
$r_1,\dots,r_p$ & target sizes of Tucker decomposition\\
$r:=\max\{r_1,\dots,r_p\}$& maximum target size\\
$\BS G$ & core tensor in Tucker decomposition\\
$\BM U^\Pa1,\dots,\BM U^\Pa p$ & factor matrices in Tucker decomposition\\
$[\Ts,\Te)$ & time range of a query\\
$L:=\Te-\Ts$&length of the query range\\
\midrule
$\theta$ & threshold for hit set pruning\\
$\sqcup$ & disjoint union\\
\bottomrule
\end{tabular}
\end{small}\end{center}
\end{table}


Given a $p$-way tensor $\BS X\in\BB R^{D_1\times\cdots\times D_p}$ and target sizes $r_1,\dots,r_p$, \emph{Tucker decomposition} \cite{tuckerals} aims to find a \emph{core tensor} $\BS G\in\BB R^{r_1\times\cdots\times r_p}$ and column-orthonormal \emph{factor matrices} $\BM U^\Pa n\in\BB R^{D_n\times r_n}$ ($n=1,\dots,p$) that minimize
\begin{equation}
\|\BS G\times_1\BM U^\Pa1\cdots\times_p\BM U^\Pa p-\BS X\|_\text F^2.\label{eq:tuckerals}
\end{equation}
Tucker decomposition is a generalization of the singular value decomposition (SVD) of matrices. Similarly with SVD, a real-world tensor $\BS X$ typically has $\BS X\approx\BS G\times_1\BM U^\Pa1\cdots\times_p\BM U^\Pa p$ even for small target sizes $r_1,\dots,r_p$ \cite{liu2012tensor}. Hence, Tucker decomposition can serve as a compressed representation of a large 
tensor. \change{
Factor matrices and the core tensor can serve as the input for downstream data mining algorithms. For example, we can apply clustering \cite{zhao2005tricluster} or anomaly detection \cite{koutra2012tensorsplat} to the row vectors of the factor matrices $\BM U^{\Pa n}$ to discover similarity and dissimilarity patterns in the data. 
}

A classic approach to Tucker decomposition is Tucker's alternating least squares method ({Tucker-ALS}) \cite{tuckerals}. At each iteration, {Tucker-ALS} optimizes the factor matrix of only one mode while fixing all other factor matrices. Tucker \cite{tuckerals} shows that the optimal factor matrix $\BM U^\Pa n$ for each mode $n$ is the $r_n$ leading left singular vectors of the matrix
\begin{align}
&\Mat n(\BS X\times_1\BM U^{\Pa1\Tp}\cdots\times_{n-1}\BM U^{\Pa{n-1}\Tp}\times_{n+1}\BM U^{\Pa{n+1}\Tp}\cdots\times_p\BM U^{\Pa p\Tp}),\label{eq:als-mat}
\end{align}
and that the optimal core tensor $\BS G$ is 
\begin{equation}
\BS X\times_1\BM U^{\Pa1\Tp}\cdots\times_p\BM U^{\Pa p\Tp}.
\end{equation}


\section{Problem Definition}
In this section, we first introduce the problem definition and then present the design goals of \Ours{}.

A tensor time series is a tensor where one of the modes represents time. Without loss of generality, let the first mode be the \emph{temporal} mode. Let $\BS X\in\BB R^{T\times D_2\times\cdots\times D_{p}}$ be a $p$-way 
tensor time series, where the size of the temporal mode is $D_1:=T$, the sizes of non-temporal modes are $D_2,\dots,D_{p}$, and the number of modes is $p\ge2$. We call $T$ the \emph{timespan}. For tensors $\BS Y_1,\dots,\BS Y_s$ of the same shape except for the temporal mode, let 
${\left[\begin{smallmatrix}\BS Y_1\\\vdots\\\BS Y_s\end{smallmatrix}\right]}$
denote concatenation along the temporal mode. A \emph{tensor stream} $\BS X\in\BB R^{*\times D_2\times\cdots\times D_{p}}$ is a tensor time series with a growing temporal mode: at each time $T$, a tensor slice $\BS X_T\in\BB R^{D_2\times\cdots\times D_p}$ is observed and appended to the tensor stream. 

\begin{figure}[t]
\begin{center}
\includegraphics[width=\linewidth]{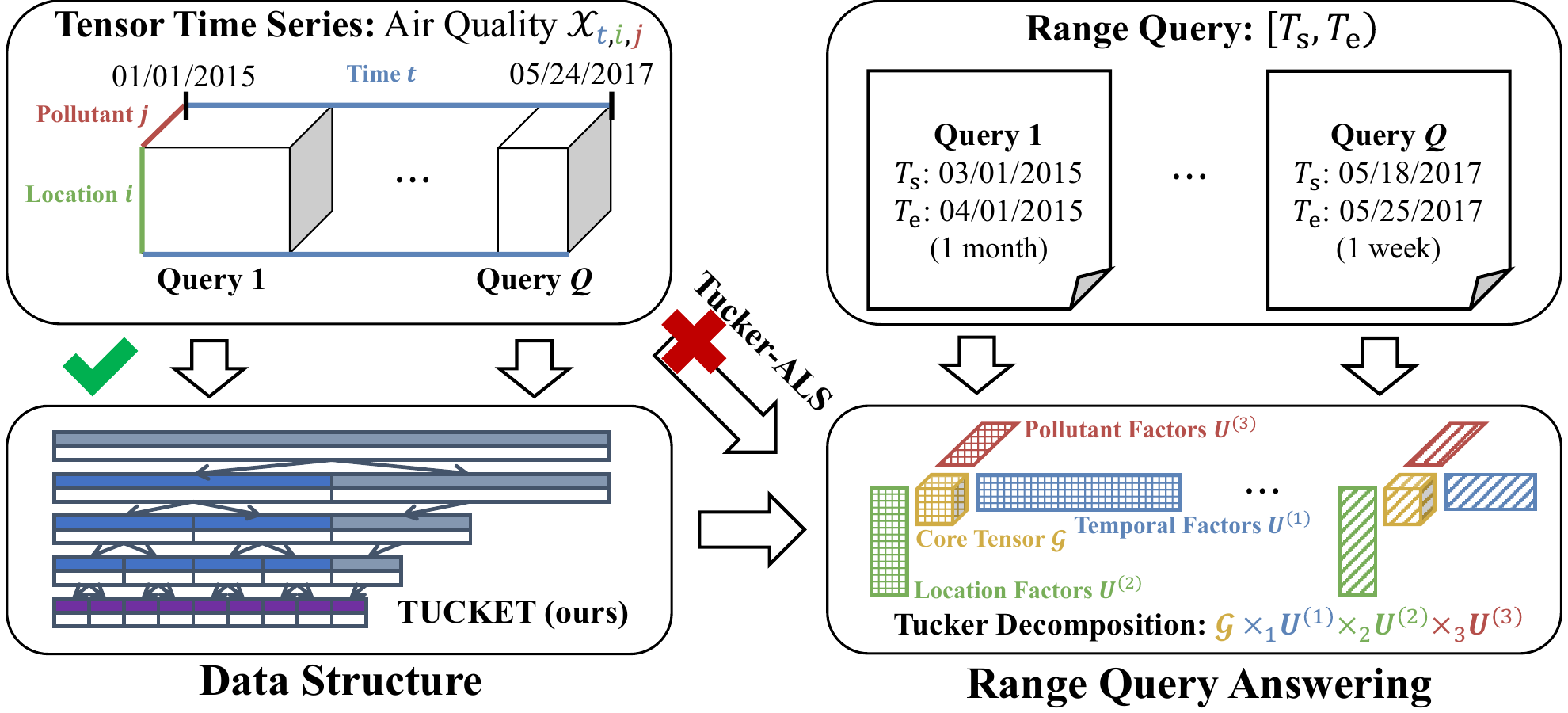}
\caption{Illustration of range queries of Tucker decomposition. It is inefficient to directly apply Tucker-ALS for each range query from scratch. Instead, we aim to design a data structure that can efficiently and accurately answer range queries of Tucker decomposition without re-computing from scratch for each query.}
\label{fig:example_rqa}
\end{center}
\end{figure}

Given a tensor stream $\BS X\in\BB R^{*\times D_2\times\cdots\times D_{p}}$ and the target sizes $(r_1,\dots,r_p)$ for Tucker decomposition, we aim to design a data structure that supports the following two operations. $T$ denotes the timespan before the operation.
\begin{itemize}[noitemsep,topsep=0pt]
\item\textbf{Range query of Tucker decomposition:} Given a time range $[\Ts,\Te)\subseteq[0,T)$, we need to efficiently compute the Tucker decomposition of the subtensor $\BS X_{[\Ts,\Te)}$ using the data structure. The output is a core tensor $\BS G\in\BB R^{r_1\times\cdots\times r_p}$ and factor matrices $\BM U^\Pa1\in\BB R^{(\Te-\Ts)\times r_1}$, $\BM U^\Pa2\in\BB R^{D_2\times r_2}$, \ldots, $\BM U^\Pa p\in\BB R^{D_p\times r_p}$.
\item\textbf{Stream update:} Given a tensor slice $\BS X_T\in\BB R^{D_2\times\cdots\times D_p}$, we need to append the tensor slice to the tensor stream and update the data structure accordingly.
\end{itemize}

\change{
The problem definition is illustrated in Figure~\ref{fig:example_rqa} with Air Quality data as an example. Air Quality data is a 3-way tensor time series $\BS X\in\BB R^{T\times D_2\times D_3}$ where $\SCR X_{t,i,j}$ represents the concentration value of air pollutant $j$ in location $i$ at time $t$. Consider a case study where we want to analyze air pollution patterns in March of each year. Here, each range query is a month (March 2015, March 2016, or March 2017; see Figure~\ref{fig:exp-case-air}). With the help of Tucker decomposition range queries, we can find out which locations share similarity pollution patterns in each month by clustering the row vectors of the location factor matrix $\BM U^\Pa2$ of Tucker decomposition of each month. Results of the case study are shown in Figure~\ref{fig:exp-case-air}. See Section~\ref{subsec:exp-case-study} for detail.
}

We design our \Ours{} with the following three design goals for a tensor time series data structure. 

\Par{G1: Frequent arbitrary range queries}
We focus on the situation where queries are frequent, and we only consider online algorithms (i.e., the algorithm has to process each operation sequentially as soon as it arrives). Thus, we need to optimize the worst-case complexity of answering each single range query. Besides that, we do not assume any extra prior knowledge about the distribution of possible range queries. Hence, we focus on the worst-case complexity parameterized by: 
(i) the maximum size of non-temporal modes, $D:=\max\{D_2,\dots,D_{p}\}$; (ii) the timespan, $T$; (iii) the maximum target size, $r:=\max\{r_1,\dots,r_p\}$; and (iv) the length of the query range, $L:=\Te-\Ts$. 

\Par{G2: Periodic stream updates}
In real-world use cases, stream updates to the tensor stream are typically periodic but may not be as frequent as range queries. For instance, in the stock example in Figure~\ref{fig:example_rqa}, the tensor stream is updated in a daily basis. Hence, we allow the stream update operation to be a little more expensive than range queries. Nonetheless, we still aim to optimize the time complexity of stream updates so that it scales at most sublinearly w.r.t.\ the total size $TD_2\cdots D_p$ of the current tensor time series.


\Par{G3: Nearly linear space}
Every preprocessing-based data structure is essentially a space--time tradeoff \cite{cobham1966recognition,hellman1980cryptanalytic} in that more preprocessing leads to higher efficiency. On the one hand, if no preprocessing were allowed, it would be impossible to outperform the na\"ive algorithm that simply computes from scratch for each query. On the other hand, if unlimited preprocessing were allowed, then a trivial algorithm would be to preprocess the answers for all possible $O(T^2)$ query ranges. To rule out such trivial algorithms, we require that the space used by the preprocessing phase should be nearly linear w.r.t.\ the timespan $T$, i.e., $\TLD O(T)$. As a remark, we assume that $p=O(1)$ and $r=o(D)$ in our complexity analysis.



\begin{figure}[t]
\begin{center}
\includegraphics[width=\linewidth]{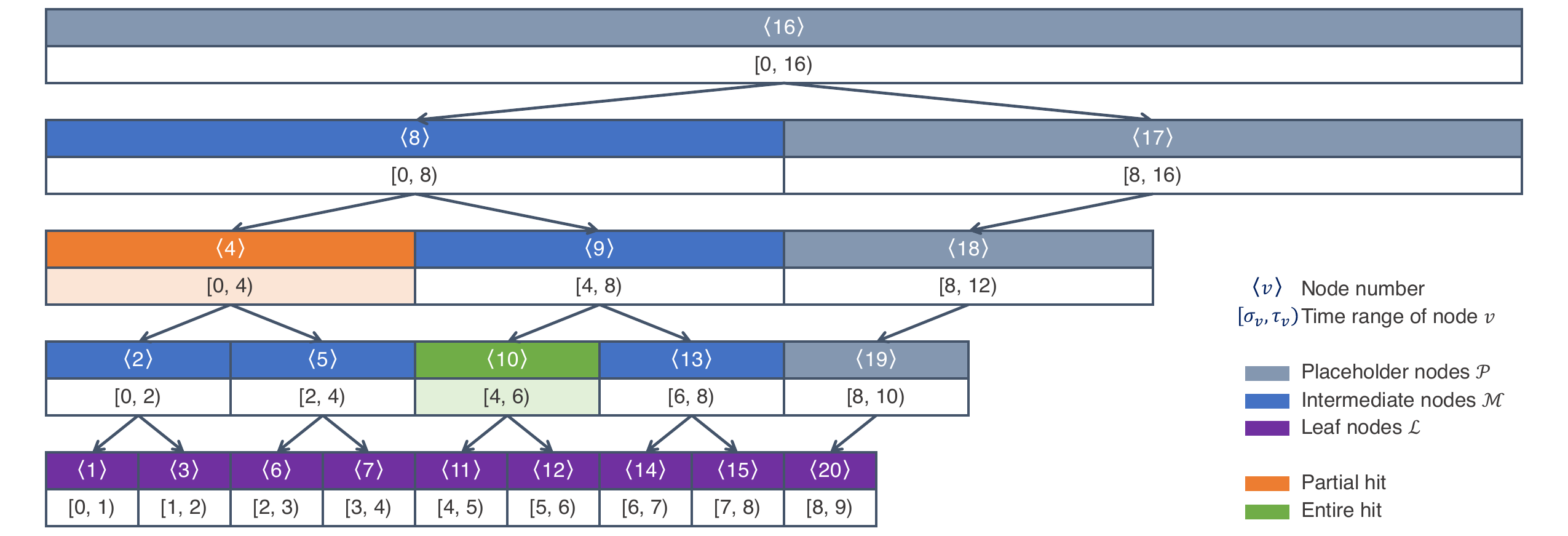}
\definecolor{EntireHit}{RGB}{113,173,71}
\definecolor{PartialHit}{RGB}{237,125,50}
\caption{Illustration of our \Ours{} over timespan $[0,9)$. It has height $5=\lceil\log_29\rceil+1$. When answering a range query $[\Ts,\Te)=[1,6)$ with pruning threshold $\theta=0.7$, node $\langle4\rangle$ is a \textcolor{PartialHit}{partial hit} because $|[1,6)\cap[0,4)|=|[1,4)|\ge 0.7|[0,4)|$, and node $\langle10\rangle$ is an \textcolor{EntireHit}{entire hit} because $[4,6)\subseteq[1,6)$. 
}
\label{fig:illust}
\end{center}
\end{figure}

\section{\Ours{}: Data Structure Design}\label{sec:design}
In this section, we detail the design of our proposed data structure \OursFull{} (\Ours{}). We first introduce the challenges of the problem and our key ideas in Section~\ref{ssec:design-ideas} and then present the design of our stream segment tree in Section~\ref{ssec:design-tree}. Due to the space limit, detailed proofs are deferred to the full version.

\subsection{Challenges \& Key Ideas}\label{ssec:design-ideas}
Our \Ours{} is designed for efficient Tucker decomposition over time ranges. An essential challenge of this problem is that the Tucker decomposition operation does not form either an Abelian group or a semigroup w.r.t.\ tensor concatenation. This means that we cannot use classic data structures such as prefix sum tables, Fenwick trees \cite{fenwick1994new}, or Cartesian trees \cite{bender2000lca} to compute Tucker decomposition over time ranges. The state-of-the-art idea 
(proposed in \cite{zoomtucker}) is to divide the time range $[0,T)$ into disjoint blocks $[0,b),[b,2b),\ldots$ of equal size $b$ and preprocess the Tucker decomposition for each block. However, this 
idea suffers from an inevitable tradeoff between efficiency and accuracy: a small block size $b$ leads to inefficiency for long time ranges; a large block size causes inaccuracy for time ranges shorter than $b$.
\change{
This is because a larger block size $b$ corresponds to a coarser-grained preprocessing, which may fail to preserve finer-grained patterns that exist temporarily in time ranges shorter than $b$. For instance, if the block size $b$ corresponds to a month, then Zoom-Tucker will be likely to yield inaccurate results when querying for a week.
}
To the best of our knowledge, our work is the first data structure that addresses this challenge without the efficiency--accuracy tradeoff.

Our first key idea, aiming to address this challenge, is to divide the time range $[0,T)$ into carefully designed uneven, overlapping blocks such that every range query can be expressed as a disjoint union of a small number of ``blocks.'' Our idea leads us to the segment tree \cite{bentley1977algorithms}, a data structure from computational geometry. Crucially, in a segment tree over a timespan $[0,T)$, we can associate each node $v$ with a subrange $[\sigma_v,\tau_v)$ such that every range $[\Ts,\Te)\subseteq[0,T)$ can be expressed as a disjoint union of at most $O(\log T)$ nodes. We call these nodes the \emph{hit set} of $[\Ts,\Te)$. Then, we can answer each range query efficiently by stitching the preprocessed subtensor decompositions of the hit set.

However, the original segment tree has a static structure and thus does not support stream data. Hence, we cannot simply apply the segment tree. As our second key idea, we propose a new data structure called the \emph{stream segment tree} to support stream updates in our setting. Our key insight regarding why segment trees are static is that it is required to be a full binary tree in order to maintain a depth of $O(\log T)$. Instead, we propose to relax this requirement and allow our stream segment tree to be incomplete. To maintain a depth of $O(\log T)$, we propose extending the root instead of only extending leaf nodes like typical balanced search trees \cite{guibas1978dichromatic}. Furthermore, each stream update operation only involves amortized $O(1)$ nodes to be updated. We will describe the detailed design of our stream segment tree in Section~\ref{ssec:design-tree}.

\subsection{Stream Segment Tree}\label{ssec:design-tree}
The basic structure of our \Ours{} is a \emph{stream segment tree}, which is a generalization of the segment tree \cite{bentley1977algorithms} from computational geometry. 
Here, we detail the design of our stream segment tree.

As we have discussed, the original segment tree does not support stream updates. To enable stream updates, our key idea is to employ an expanded segment tree with an incomplete structure that reserves the position of future nodes to support efficient updates but does not construct them explicitly. 

Specifically, a stream segment tree is a binary tree where each node $v$ is associated with a time range $[\sigma_v,\tau_v)$. There are three types of nodes in a stream segment tree:
\emph{leaf} nodes $\CAL L$, \emph{intermediate} nodes $\CAL M$, and \emph{placeholder} nodes $\CAL P$. 

\Par{Leaf nodes}
\change{Each leaf node $v\in\CAL L$ represents a tensor slice $\BS X_t$, so we let the leaf range be $[\sigma_v,\tau_v):=[t,t+1)$.} After $T$ updates, we require time ranges of leaf nodes to be \emph{contiguous}, i.e., the whole range $[0,T)$ is a disjoint union of the time ranges of leaf nodes:
\AL{\bigsqcup_{v\in\CAL L}[\sigma_v,\tau_v)=[0,T).}
Besides that, we preprocess the Tucker decomposition of $\BS X_{[t,t+1)}$ and store it as $\BS Y_v:=\BS H_v\times_1\BM V_v^\Pa1\cdots\times_p\BM V_v^\Pa p$. Note that we only compute the core tensor $\BS H_v$ and factor matrices $\BM V_v^\Pa1,\ldots,\BM V_v^\Pa p$ but never actually compute $\BS Y_v$, i.e., $\BS Y_v$ is only a symbol to refer to the product. 

\Par{Intermediate nodes}
For each intermediate node $v\in\CAL M$ with time range $[\sigma_v,\tau_v)$, it represents a subtensor $\BS X_{[\sigma_v,\tau_v)}$. An intermediate node has exactly two children nodes $u_1,u_2\in\CAL L\cup\CAL M$ that have \emph{adjoint} time ranges: 
\AL{[\sigma_v,\tau_v):=[\sigma_{u_1},\tau_{u_1})\sqcup[\sigma_{u_2},\tau_{u_2}).}
Besides that, similarly with leaf nodes, we preprocess the Tucker decomposition of $\BS X_{[\sigma_v,\tau_v)}$ and store it as $\BS Y_v:=\BS H_v\times_1\BM V_v^\Pa1\cdots\times_p\BM V_v^\Pa p$. Here, $\BS Y_v$ is still only a symbol to refer to the product. 

\Par{Placeholder nodes}
A placeholder node $v\in\CAL P$ represents a future subtensor where some of its slices have not been observed yet. Formally, if the current observed timespan is $[0,T)$. then the time range $[\sigma_v,\tau_v)$ of the placeholder node has $\sigma_v<T$ and $\tau_v\ge T$. A placeholder node has either one or two children. If $v$ has only one child, then it has a left child $u_1\in\CAL L\cup\CAL M\cup\CAL P$; otherwise, $v$ has a left child $u_1\in\CAL L\cup\CAL M$ and a right child $u_2\in\CAL P$. Similarly with intermediate nodes, we require its children to have \emph{adjoint} time ranges, i.e., $[\sigma_v,\tau_v):=[\sigma_{u_1},\tau_{u_1})\sqcup[\sigma_{u_2},\tau_{u_2})$. Meanwhile, unlike leaf and intermediate nodes, since the subtensor of the time range $[\sigma_v,\tau_v)$ has not been completely observed yet, we do not preprocess the Tucker decomposition of a placeholder node and do not allow placeholder nodes to be in the hit set.

\Par{Logarithmic height}
To efficiently answer range queries, we want that every range query $[\Ts,\Te)\subseteq[0,T)$ can be divided into a disjoint union of a small number of nodes (called the \emph{hit set}) in the stream segment tree. We make a key observation on the relation between the size of the hit set and the height of the stream segment tree, as formally stated in Lemma~\ref{LEM:hit-height}.

\begin{lemma}[Hit set v.s.\ height]\label{LEM:hit-height}
Given a stream segment tree of height $h\ge1$, for every range query, there exists a hit set of size $\le\max\{2(h-1),1\}$. 
\end{lemma}

\begin{proof}[\change{Proof sketch}]
\change{
First, if the query range $[\Ts,\Te)$ is a prefix or a suffix of the time range $[\sigma_v,\tau_v)$ of a node $v$ (i.e., $\Ts=\sigma_v$ or $\Te=\tau_v$), then an induction shows that there exists a hit set of size $\le h$. Next, if the query range is neither a prefix nor a suffix of the time range of any node, then we can show that there exists two non-root nodes $v_1,v_2$ such that $\tau_{v_1}=\sigma_{v_2}$ and that $[\Ts,\Te)=[\sigma_{v_1},\tau_{v_1})\sqcup[\sigma_{v_2},\tau_{v_2})$. In this case, $[\Ts,\Te)\cap[\sigma_{v_1},\tau_{v_1})$ is a suffix of $[\sigma_{v_1},\tau_{v_1})$, and $[\Ts,\Te)\cap[\sigma_{v_2},\tau_{v_2})$ is a prefix of $[\sigma_{v_2},\tau_{v_2})$. Since $v_1,v_2$ are not the root, then $[\Ts,\Te)$ has a hit set of size $\le2(h-1)$. Together, the size of the hit set is $\le\max\{2(h-1),h\}=\max\{2(h-1),1\}$.
}
\end{proof}

Lemma~\ref{LEM:hit-height} motivates us to require the stream segment tree to have a small height. We show that the stream segment tree can indeed have a logarithmic height w.r.t.\ the time range $T$, as formally stated in Theorem~\ref{THM:height-T}.

\begin{theorem}[Logarithmic height]\label{THM:height-T}
There exists a stream segment tree structure over range $[0,T)$ of height $\le\lceil\log_2T\rceil+1$.
\end{theorem}

\begin{proof}[\change{Proof sketch}]
\change{Using the algorithm in Section~\ref{ssec:op-append} to append the tensor slices $\BS X_0,\dots,\BS X_{T-1}$ one by one, we can build a stream segment tree over $[0,T)$. By Theorem~\ref{THM:alg-log-height}, this stream segment tree has height $\lceil\log_2T\rceil+1$.}
\end{proof}

The tree structure in Theorem~\ref{THM:height-T} not only exists in theory but can also be maintained efficiently. We will present an efficient algorithm to maintain the logarithmic height 
in a stream update in Section~\ref{ssec:op-append}. From now on, we will refer to the tree structure in Theorem~\ref{THM:height-T} simply as the ``stream segment tree.''

\begin{proposition}[Space complexity]
\label{PROP:space_complexity}
The space complexity of a stream segment tree over range $[0,T)$ is $O((rD+r^p)T+rT\log T)$.
\end{proposition}
\begin{proof}[\change{Proof sketch}]
\change{In a stream segment tree over the range $[0,T)$, there are $O(T)$ nodes each of which has $p-1$ non temporal factor matrices of the size $O(rD)$ and a core tensor of the size $O(r^p)$.
In addition, at each level of the tree, the sum of the sizes of temporal factor matrices is $O(rT)$.
Therefore, the space complexity of the stream segment tree is $O((prD+r^p)T + rT\log T)$.}
\end{proof}

\Par{Example}
An example of our stream segment tree over timespan $[0,9)$ is illustrated in Figure~\ref{fig:illust}. It has height $5=\lceil\log_29\rceil+1$. 
\change{Leaf nodes $\langle1\rangle,\langle3\rangle,\langle6\rangle,\langle7\rangle,\langle11\rangle,\langle12\rangle,\langle14\rangle,\langle15\rangle,\langle20\rangle$ represent tensor slices $\BS X_0,\dots,\BS X_8$, respectively. Intermediate nodes store preprocessed results.
We will also use this example in the subsequent sections to illustrate other operations.}

\section{\Ours{}: Core Algorithms}\label{sec:algo}
In this section, we present two core algorithms that will be used in the operations of \Ours{} in Section~\ref{sec:op}. We describe how to optimally prune the hit set in Section~\ref{ssec:core-prune} and how to stitch subtensor decompositions in the hit set in Section~\ref{ssec:core-stitch}.

\subsection{Optimally Pruning the Hit Set}\label{ssec:core-prune}
The first core algorithm of \Ours{} is finding a small hit set for each range query. By Lemma~\ref{LEM:hit-height} \& Theorem~\ref{THM:height-T}, we have shown that the hit set has a small size $O(\log T)$. Although this size cannot be improved for general operations, here we present a key observation about the Tucker decomposition of a subtensor and leverage this observation to further prune the hit set.

\Par{Approximating a subtensor decomposition}
Here we present our key observation about the Tucker decomposition of a subtensor. Suppose that a tensor $\BS X_{[\sigma_v,\tau_v)}$ has Tucker decomposition $\BS H_v\times_1\BM V_v^\Pa1\cdots\times_p\BM V_v^\Pa p$. Due to the low-rank nature of real-world tensors \cite{liu2012tensor}, they can typically be well approximated by Tucker decomposition:
\AL{\BS X_{[\sigma_v,\tau_v)}\approx\BS H_v\times_1\BM V_v^\Pa1\cdots\times_p\BM V_v^\Pa p.}
We observe that for real-world tensors, there typically exists a threshold $0<\theta<1$ (see Section~\ref{subsec:exp-sensitivity}) such that for a sub-range $[\Ts',\Te')\subseteq[\sigma_v,\tau_v)$ with $\big|[\Ts',\Te')\big|\ge\theta\big|[\sigma_v,\tau_v)\big|$ (i.e., the sub-range $[\Ts',\Te')$ is not too small compared with $[\sigma_v,\tau_v)$), the subtensor $\BS X_{[\Ts',\Te')}$ can be well approximated by 
\AL{\BS X_{[\Ts',\Te')}&=(\BS X_{[\sigma_v,\tau_v)})_{[\Ts'-\sigma_v,\Te'-\sigma_v)}\\
&\approx(\BS H_v\times_1\BM V_v^\Pa1\times_2\BM V_v^\Pa2\cdots\times_p\BM V_v^\Pa p)_{[\Ts'-\sigma_v,\Te'-\sigma_v)}\\
&=\BS H_v\times_1(\BM V_v^\Pa1)_{[\Ts'-\sigma_v,\Te'-\sigma_v)}\times_2\BM V_v^\Pa2\cdots\times_p\BM V_v^\Pa p
.}
This almost yields an approximate Tucker decomposition of $\BS X_{[\Ts',\Te')}$, except that the temporal factor matrix $(\BM V_v^\Pa1)_{[\Ts'-\sigma_v,\Te'-\sigma_v)}$ is not necessarily column-orthonormal. To make it column-orthonormal, we can first compute a QR decomposition \cite{gander1980algorithms} $(\BM V_v^\Pa1)_{[\Ts'-\sigma_v,\Te'-\sigma_v)}=:\TLD{\BM Q}\TLD{\BM R}$ (where $\TLD{\BM Q}$ is column-orthonormal) and then use the reverse associativity\footnote{\label{footnote:reverse_asso}The reverse associativity means that $\BS Z\times_n\BM A\times_n\BM B=\BS Z\times_n(\BM B\BM A)$.}~\cite{kolda2009tensor} of $\times_1$ to give a Tucker decomposition:
\AL{
\BS X_{[\Ts',\Te')}&\approx\BS H_v\times_1(\BM V_v^\Pa1)_{[\Ts'-\sigma_v,\Te'-\sigma_v)}\times_2\BM V_v^\Pa2\cdots\times_p\BM V_v^\Pa p\\
&=\BS H_v\times_1(\TLD{\BM Q}\TLD{\BM R})\times_2\BM V_v^\Pa2\cdots\times_p\BM V_v^\Pa p\\
&=(\BS H_v\times_1\TLD{\BM R})\times_1\TLD{\BM Q}\times_2\BM V_v^\Pa2\cdots\times_p\BM V_v^\Pa p\label{eq:partial-tucker}
,}
where the core tensor is $\BS H_v\times_1\TLD{\BM R}$, and the temporal factor matrix is $\TLD{\BM Q}$. In this way, we can efficiently compute an approximate Tucker decomposition of a subtensor $\BS X_{[\Ts',\Te')}$ using only a QR decomposition and a mode-$1$ product and do not need to further divide $[\Ts',\Te')$ into smaller sub-ranges. This helps to reduce the size of the hit set.

\Par{Formulation of hit set pruning}
This key observation motivates us to consider \emph{partial hits}. Let $0<\theta<1$ denote the threshold 
above. We call a node $v\in\CAL M$ a \emph{partial hit} of a range query $[\Ts,\Te)$ if $\big|[\Ts,\Te)\cap[\sigma_v,\tau_v)\big|\ge\theta\big|[\sigma_v,\tau_v)\big|$ and $[\sigma_v,\tau_v)\not\subseteq[\Ts,\Te)$; we call it an \emph{entire hit} of $[\Ts,\Te)$ if $[\sigma_v,\tau_v)\subseteq[\Ts,\Te)$. Using Eq.~\eqref{eq:partial-tucker}, we can efficiently approximate the Tucker decomposition of $\BS X_{[\Ts,\Te)\cap[\sigma_v,\tau_v)}$. Hence, we can reduce the size of the hit set by allowing partial hits in the hit set. Formally, minimizing the size of the hit set $\CAL S$ can be formulated as the following optimization problem:
\AL{\min_{\CAL S\subseteq\CAL L\cup\CAL M}&|\CAL S|,\label{eq:prune-obj}\\
\text{s.t.\quad\,}&[\Ts,\Te)=\textstyle\bigsqcup\limits_{v\in\CAL S}[\Ts,\Te)\cap[\sigma_v,\tau_v),\label{eq:prune-c1}\\
&\big|[\Ts,\Te)\cap[\sigma_v,\tau_v)\big|\ge\theta\big|[\sigma_v,\tau_v)\big|,\;\;\forall v\in\CAL S\label{eq:prune-c2}
.}

\ALG{recall}{(\textsc{Recall}): Finding a pruned hit set
}{
\REQUIRE{current node $v$; query range $[\Ts,\Te)$}
\ENSURE{a pruned hit set of $[\Ts,\Te)$ in the subtree rooted at $v$}
\IF{$v\in\CAL L\cup\CAL M$ \textbf{and} $\big|[\Ts,\Te)\cap[\sigma_v,\tau_v)\big|\ge\theta\big|[\sigma_v,\tau_v)\big|$}
    \RETURN$\{v\}$
\ENDIF
\STATE let $u_1,u_2$ be the left and right children of $v$, respectively
\IF{$\Te\le\tau_{u_1}$}
    \RETURN$\textsc{Recall}(u_1,[\Ts,\Te))$
\ELSIF{$\Ts\ge\sigma_{u_2}$}
    \RETURN$\textsc{Recall}(u_2,[\Ts,\Te))$
\ELSE
    \RETURN$\textsc{Recall}(u_1,[\Ts,\Te))\cup\textsc{Recall}(u_2,[\Ts,\Te))$
\ENDIF
}

\Par{An optimal algorithm for pruning}
To solve the formulation above for hit set pruning, we propose an efficient recursive algorithm that runs in $O(\log T)$ time. The basic idea is as follows. We start from the root node. If the root node is a partial hit, then we stop and return the root node as the hit set. Otherwise, we consider its two children and repeat the procedure above. The overall procedure is presented in \ALGref{recall}. Since the height of the stream segment tree is $O(\log T)$, and \ALGref{recall} visits at most two nodes at each height, then the total running time of \ALGref{recall} is $O(\log T)$.

Furthermore, our Theorem~\ref{THM:prune-opt} shows that our \ALGref{recall} is indeed optimal --- it can find the smallest hit set that satisfies the constraints Eqs.~\eqref{eq:prune-c1} \& \eqref{eq:prune-c2}.

\begin{theorem}[Optimality \& complexity of \ALGref{recall}]\label{THM:prune-opt}
Given a range query $[\Ts,\Te)$, \ALGref{recall} minimizes the formulation in Eq.~\eqref{eq:prune-obj} within $O(\log T)$ running time and finds a hit set with $O(1)$ partial hits and $O(\log L)$ entire hits, where $L:=\Te-\Ts$. 
\end{theorem}

\begin{proof}[\change{Proof sketch}]
\change{\emph{Optimality.} Note that if the two children of a node $v$ are both in the hit set, then replacing the two children with the node $v$ gives a smaller, valid hit set. Using this fact, it can be shown that every node in the optimal hit set should not have a parent node that is also a valid hit. Finally, by analyzing the top-down procedure of \ALGref{recall}, we can show that \ALGref{recall} can indeed find such a hit set, and that the hit set cannot be replaced with a smaller hit set.
}

\change{\emph{Complexity.} An induction similar with the proof of Lemma~\ref{LEM:hit-height} shows that the hit set found by \ALGref{recall} has $O(1)$ partial hits and $O(\log L)$ entire hits. Since the height of the stream segment tree is $O(\log T)$, the number of nodes traversed in the process of finding the hit set is at most $O(\log T+\log L)=O(\log T)$. Finally, as \ALGref{recall} performs $O(1)$ operations per node traversed, its time complexity is $O(\log T)$.
}
\end{proof}

\Par{Example} \ALGref{recall} is exemplified in Figure~\ref{fig:illust}. When answering a range query $[1,6)$ with $\theta=0.7$, $[0,4)$ is a partial hit because $|[1,6)\cap[0,4)|\ge 0.7|[0,4)|$, and $[4,6)$ is an entire hit because $[4,6)\subseteq[1,6)$. For the partial hit $[0,4)$, we use Eq.~\eqref{eq:partial-tucker} to approximate the Tucker decomposition of the sub-range $[1,6)\cap[0,4)=[1,4)$.

\subsection{Stitching Subtensor Decompositions}\label{ssec:core-stitch}
Another core algorithm of \Ours{} is stitching subtensor decompositions in the hit set. Given a range query $[\Ts,\Te)$, suppose that the (pruned) hit set is $\CAL S=\{v_1,\dots,v_s\}$, where $s:=|\CAL S|$ denotes the size of the hit set. For each partial hit, we use Eq.~\eqref{eq:partial-tucker} to compute its approximate Tucker decomposition $\TLD{\BS Y}_i$ of the subtensor $\BS X_{[\Ts,\Te)\cap[\sigma_{v_i},\tau_{v_i})}$; for each entire hit $v_i$, we retrieve its preprocessed Tucker decomposition $\TLD{\BS Y}_i:=\BS Y_{v_i}$. Same as before, here $\TLD{\BS Y}_1,\dots,\TLD{\BS Y}_s$ are just symbols to refer to the Tucker decomposition products. We aim to efficiently compute an approximate Tucker decomposition of $\BS X_{[\Ts,\Te)}$ using these preprocessed subtensor decompositions $\TLD{\BS Y}_1,\dots,\TLD{\BS Y}_s$. 

A key observation is that $\BS X_{[\Ts,\Te)\cap[\sigma_{v_i},\tau_{v_i})}\approx\TLD{\BS Y}_i$ due to the low-rank nature of real-world tensors \cite{liu2012tensor}. This motivates us to express $\BS X_{[\Ts,\Te)}$ as a concatenation of the hit set along the temporal mode:
\AL{\BS X_{[\Ts,\Te)}=\MAT{\BS X_{[\Ts,\Te)\cap[\sigma_{v_1},\tau_{v_1})}\\\vdots\\\BS X_{[\Ts,\Te)\cap[\sigma_{v_s},\tau_{v_s})}}\approx\MAT{\TLD{\BS Y}_1\\\vdots\\\TLD{\BS Y}_s}.}
Next, we design an efficient algorithm to compute the Tucker decomposition of $\BS X_{[\Ts,\Te)}$ by stitching the subtensor decompositions $\TLD{\BS Y}_1,\dots,\TLD{\BS Y}_s$. The key idea here is to leverage the concatenation form of $\TLD{\BS Y}:=\left[\begin{smallmatrix}\TLD{\BS Y}_1\\\vdots\\\TLD{\BS Y}_s\end{smallmatrix}\right]$ and again utilize the reverse associativity\textsuperscript{\ref{footnote:reverse_asso}}~\cite{kolda2009tensor} 
of the tensor--matrix product so as to efficiently compute the matricizations in Tucker-ALS.

Let $\TLD{\BS H}_i$ and $\TLD{\BM V}_i^\Pa1,\dots,\TLD{\BM V}_i^\Pa p$ denote the core tensor and the factor matrices in $\TLD{\BS Y}_i$, respectively, and let $\BS G$ and $\BM U^\Pa1,\dots,\BM U^\Pa p$ denote the core tensor and the factor matrices of $\BS X_{[\Ts,\Te)}$ to be computed, respectively. Since the optimal update of the factor matrix $\BM U^\Pa n$ is the $r_n$ leading left singular vectors of the matricization in Eq.~\eqref{eq:als-mat}, we need to compute this matricization efficiently. Since the concatenation is along the temporal mode, we will describe how to efficiently compute the matricization for the temporal mode and the non-temporal modes separately. The overall procedure of stitching subtensor decompositions is presented in \ALGref{stitch}.

\Par{Matricization of the temporal mode}
Our goal is to compute the matricization in Eq.~\eqref{eq:als-mat} without explicitly computing the large tensor $\TLD{\BS Y}$. First, we rewrite the matricizations of $\TLD{\BS Y}_i$ via the matrix Kronecker product $\otimes$:
\AL{\Mat1(\TLD{\BS Y}_i)&=\Mat1(\TLD{\BS H}_i\times_1\TLD{\BM V}_i^\Pa1\cdots\times_p\TLD{\BM V}_i^\Pa p)\\
&=\TLD{\BM V}_i^\Pa1\Mat1(\TLD{\BS H}_i)\bigotimes_{m=2}^p\TLD{\BM V}_i^{\Pa m\Tp}.}
Similarly, we can rewrite the matricization in Eq.~\eqref{eq:als-mat} as
\AL{\Mat 1(\TLD{\BS Y}\times_1\BM U^{\Pa2\Tp}\cdots\cdots\times_p\BM U^{\Pa p\Tp})=\Mat1(\TLD{\BS Y})\bigotimes_{m=2}^p\BM U^\Pa m.\label{eq:stitch-temp-1}}
Since the matricization of the concatenation $\TLD{\BS Y}$ is equal to the concatenation of $\Mat1(\TLD{\BS Y}_i)$'s, then by the mixed-product property\footnote{The mixed-product property means that $(\BM A\otimes\BM B)(\BM C\otimes\BM D)=(\BM A\BM C)\otimes(\BM B\BM D)$.} \cite{broxson2006kronecker} 
of the Kronecker product, Eq.~\eqref{eq:stitch-temp-1} can be further rewritten as:
\AL{&\left[\begin{smallmatrix}\TLD{\BM V}_1^\Pa1\Mat1(\TLD{\BS H}_1)(\bigotimes_{m=2}^p\TLD{\BM V}_1^{\Pa m\Tp})(\bigotimes_{m=2}^p\BM U^\Pa m)\\\vdots\\\TLD{\BM V}_s^\Pa1\Mat1(\TLD{\BS H}_s)(\bigotimes_{m=2}^p\TLD{\BM V}_s^{\Pa m\Tp})(\bigotimes_{m=2}^p\BM U^\Pa m)\end{smallmatrix}\right]\\
={}&\left[\begin{smallmatrix}\TLD{\BM V}_1^\Pa1\Mat1(\TLD{\BS H}_1)\bigotimes_{m=2}^p(\TLD{\BM V}_1^{\Pa m\Tp}\BM U^\Pa m)\\\vdots\\\TLD{\BM V}_s^\Pa1\Mat1(\TLD{\BS H}_s)\bigotimes_{m=2}^p(\TLD{\BM V}_s^{\Pa m\Tp}\BM U^\Pa m)\end{smallmatrix}\right]\\
={}&\left[\begin{smallmatrix}\Mat1(\TLD{\BS H}_1\times_1\TLD{\BM V}_1^\Pa1\times_2(\BM U^{\Pa2\Tp}\TLD{\BM V}_1^\Pa2)\cdots\times_p(\BM U^{\Pa p\Tp}\TLD{\BM V}_1^\Pa p))\\\vdots\\\Mat1(\TLD{\BS H}_s\times_1\TLD{\BM V}_s^\Pa1\times_2(\BM U^{\Pa2\Tp}\TLD{\BM V}_s^\Pa2)\cdots\times_p(\BM U^{\Pa p\Tp}\TLD{\BM V}_s^\Pa p))\end{smallmatrix}\right].\label{eq:stitch-temp-2}}
Computintg Eq.~\eqref{eq:stitch-temp-2} only involves small matrices for non-temporal modes and avoids explicitly computing the large tensor $\TLD{\BS Y}$ which requires $O(rD^{p-1}(T_e-T_s))$. 
\begin{lemma}[Time complexity of the temporal mode]\label{LEM:time_complexity_temporal}
Computing the matricization of the temporal mode in Eq.~\eqref{eq:stitch-temp-2} takes $O((r^2D+r^{p+1})s+r^{p}L)$ time where $L=\Te-\Ts$.
\end{lemma}

\ALG{stitch}{(\textsc{Stitch}): Stitching subtensor decompositions}{
\REQUIRE{subtensor decompositions $\TLD{\BS Y}_i:=\TLD{\BS H}_i\times_1\TLD{\BM V}_i^\Pa1\cdots\times_p\TLD{\BM V}_i^\Pa p$ of the hit set $\{v_1,\dots,v_s\}$ 
}
\ENSURE{stitched Tucker decomposition of $\left[\begin{smallmatrix}\TLD{\BS Y}_1\\\vdots\\\TLD{\BS Y}_s\end{smallmatrix}\right]$}
\STATE randomly initialize $\BM U^\Pa2,\dots,\BM U^\Pa p$
\REPEAT 
    \STATE obtain $\BM Z^\Pa1$ using Eq.~\eqref{eq:stitch-temp-2} \label{alg:line:stitch_start}
    \STATE let $\BM U^\Pa1$ be the $r_1$ leading left singular vectors of $\BM Z^\Pa1$
    \STATE$t_0\gets0$
    \FOR{$i\gets1,\dots,s$}
        \STATE$t_i\gets t_{i-1}+(\tau_{v_i}-\sigma_{v_i})$
    \ENDFOR
    \FOR{$n\gets2,\dots,p$}
        \STATE obtain $\BM Z^\Pa n$ using Eq.~\eqref{eq:stitch-nontemp-3}
        \STATE let $\BM U^\Pa n$ be the $r_n$ leading left singular vectors of $\BM Z^\Pa n$
    \ENDFOR
    \STATE reshape $\BM Z^\Pa p$ into a tensor $\BS Z^\Pa p\in\BB R^{r_1\times\cdots\times r_{p-1}\times D_p}$
    \STATE$\BS G\gets\BS Z^\Pa p\times_p\BM U^{\Pa p\Tp}$  \label{alg:line:stitch_end}
\UNTIL{converged}
\RETURN$(\BS G,\BM U^\Pa1,\dots,\BM U^\Pa p)$
}

\begin{proof}[\change{Proof sketch}]
\change{Eq.~\eqref{eq:stitch-temp-2} consists of three computations whose costs are as follows:
for $n = 2,...,p$ and $i = 1,...,s$,
(1) matrix multiplications $\BM U^{\Pa n\Tp}\TLD{\BM V}_i^\Pa n$, (2) tensor-matrix products between $\TLD{\BS H}_i$ and the preceding results $\BM U^{\Pa n\Tp}\TLD{\BM V}_i^\Pa n$, and 
(3) tensor-matrix products between the preceding results and matrices $\TLD{\BM V}_i^\Pa 1$ take $O(r^2Ds)$, $O(r^{p+1}s)$, and $O(r^{p}L)$ time, respectively.
Therefore, the complexity for computing Eq.~\eqref{eq:stitch-temp-2} is $O((r^2D+r^{p+1})s +r^{p}L)$.}
\end{proof}


\Par{Matricization of the non-temporal modes}
For a non-temporal mode $n=2,\dots,p$, the matricization is different from Eq.~\eqref{eq:stitch-temp-2} because the concatenation is along the temporal mode. Nevertheless, we can still consider using the Kronecker product to rewrite the matricization in Eq.~\eqref{eq:als-mat} as:
\AL{\Mat n(\TLD{\BS Y})\bigotimes_{m\ne n}\BM U^\Pa m=[\Mat n(\TLD{\BS Y}_1),\ldots,\Mat n(\TLD{\BS Y}_s)]\bigotimes_{m\ne n}\BM U^\Pa m.\label{eq:stitch-nontemp-1}}
Let $t_0:=0$, and $t_i:=t_{i-1}+(\tau_{v_i}-\sigma_{v_i})$ for $i=1,\dots,s$. Then, each hit node $v_i$ corresponds to the subtensor $\TLD{\BS Y}_{[t_{i-1},t_i)}$. By splitting the temporal factor matrix as $\BM U^\Pa1=\left[\begin{smallmatrix}\BM U^\Pa1_{[t_0,t_1)}\\\vdots\\\BM U^\Pa1_{[t_{s-1},t_s)}\end{smallmatrix}\right]$ and using the mixed-product property of the Kronecker product again, we can further rewrite Eq.~\eqref{eq:stitch-nontemp-1} as
\AL{&\!\!\sum_{i=1}^s\Mat n(\TLD{\BS Y}_i)\bigg(\BM U^\Pa n_{[t_{i-1},t_i)}\otimes\bigotimes_{m\ne1,n}\BM U^\Pa m\bigg)\\
={}&\!\!\sum_{i=1}^s\TLD{\BM V}_i^\Pa n\Mat n(\TLD{\BS H}_i)\bigg(\bigotimes_{m\ne n}\TLD{\BM V}_i^{\Pa m\Tp}\bigg)\bigg(\BM U^\Pa1_{[t_{i-1},t_i)}\otimes\bigotimes_{m\ne1,n}\BM U^\Pa m\bigg)\\
={}&\!\!\sum_{i=1}^s\TLD{\BM V}_i^\Pa n\Mat n(\TLD{\BS H}_i)\bigg((\TLD{\BM V}_i^{\Pa1\Tp}\BM U^\Pa1_{[t_{i-1},t_i)})\otimes\bigotimes_{m\ne1,n}(\TLD{\BM V}_i^{\Pa m\Tp}\BM U^\Pa m)\bigg)
.\label{eq:stitch-nontemp-2}}
Finally, we rewrite the Kronecker product form in Eq.~\eqref{eq:stitch-nontemp-2} back to the matricization form:
\AL{\sum\limits_{i=1}^s&\Mat n(\TLD{\BS H}_i\times_1(\BM U_{[t_{i-1},t_i)}^{\Pa1\Tp}\TLD{\BM V}_i^\Pa1)\nonumber\\&\times_2(\BM U^{\Pa2\Tp}\TLD{\BM V}_i^\Pa2)\cdots\times_{n-1}(\BM U^{\Pa{n-1}\Tp}\TLD{\BM V}_i^\Pa{n-1})\times_n\TLD{\BM V}_i^\Pa n\nonumber\\&\times_{n+1}(\BM U^{\Pa{n+1}\Tp}\TLD{\BM V}_i^\Pa{n+1})\cdots\times_{p}(\BM U^{\Pa{p}\Tp}\TLD{\BM V}_i^\Pa{p})).\label{eq:stitch-nontemp-3}}
Computing Eq.~\eqref{eq:stitch-nontemp-3} only involves small matrices for non-temporal modes and avoids explicitly computing the large tensor $\TLD{\BS Y}$.

\begin{lemma}[Time complexity of the non-temporal modes]\label{LEM:time_complexity_nontemporal}
Computing the matricization of the non-temporal modes in Eq.~\eqref{eq:stitch-nontemp-3} takes $O(r^{p}Ds + r^{2}L)$ time where $L:=\Te-\Ts$.
\end{lemma}

\begin{proof}[\change{Proof sketch}]
\change{Eq.~\eqref{eq:stitch-nontemp-3} needs four computations:
for $m = 2,...,n-1,n+1,...,p$ and $i = 1,...,s$,
(1) matrix multiplications $\BM U_{[t_{i-1},t_i)}^{\Pa1\Tp}\TLD{\BM V}_i^\Pa1$,
(2) matrix multiplications $\BM U^{\Pa m\Tp}\TLD{\BM V}_i^\Pa m$, (3) tensor-matrix products between $\TLD{\BS H}_i$ and the current results $\BM U^{\Pa m\Tp}\TLD{\BM V}_i^\Pa m$, and 
(4) tensor-matrix products between the current results and matrices $\TLD{\BM V}_i^\Pa n$ take $O(r^{2}L)$, $O(r^2Ds)$, $O(r^{p+1}s)$, and $O(r^{p}Ds)$ time, respectively.
Hence, the total complexity for computing Eq.~\eqref{eq:stitch-nontemp-3} is $O(r^{p}Ds+ r^{2}L)$.}
\end{proof}

\Par{\change{Error analysis}}
\change{
We provide an error analysis of our \textsc{Stitch} algorithm in the following Proposition~\ref{PRP:err}.
}

\begin{proposition}[\change{Error bound}]\label{PRP:err}
\change{Let $\BS X$ be the concatenation of subtensors $\BS X^\Pa i$ ($i=1,\dots,s$), and let $\BS Y^\Pa i$ denote the rank-$r$ Tucker decomposition of $\BS X^\Pa i$. Suppose that alternating least squares are solved optimally, and that $\BS X$ is approximately low-rank (i.e., $\BS X=\BS W+\BS E$ where $\BS W$ has Tucker rank $\le r$, and $\frac{\|\BS E\|_\textnormal F}{\|\BS X\|_\textnormal F}\le\epsilon$ for small $\epsilon>0$). Then, the stitching algorithm finds a rank-$r$ Tucker decomposition $\BS Y$ of $\BS X$ with reconstruction error $\frac{\|\BS X-\BS Y\|_\textnormal F}{\|\BS X\|_\textnormal F}\le O(\epsilon)$.}
\end{proposition}

\begin{proof}[\change{Proof sketch}]
\change{Let $\BS W^\Pa i$ and $\BS E^\Pa i$ denote the part of $\BS W$ and $\BS E$ corresponding to the time range of $\BS X^\Pa i$, respectively. Then, $\|\BS X^\Pa i-\BS Y^\Pa i\|_\textnormal F\le\|\BS X^\Pa i-\BS W^\Pa i\|_\text F=\|\BS E^\Pa i\|_\textnormal F$ for all $i$. Let $\TLD{\BS Y}$ denote the concatenation of $\BS Y^\Pa i$. Thus,
\AL{
\|\BS X-\TLD{\BS Y}\|_\text F&
\le\sqrt{\sum_{i=1}^s\|\BS E^\Pa i\|_\textnormal F^2}=\|\BS E\|_\textnormal F\le\epsilon\|\BS X\|_\textnormal F
.}
Since $\BS Y$ is the Tucker decomposition of $\TLD{\BS Y}$, then
\AL{
\|\BS X-\BS Y\|_\text F\le{}&\|\BS X-\TLD{\BS Y}\|_\text F+\|\TLD{\BS Y}-\BS Y\|_\text F\le\epsilon\|\BS X\|_\textnormal F+\|\TLD{\BS Y}-\BS W\|_\text F\nonumber
\\\le{}&\epsilon\|\BS X\|_\textnormal F+\|\TLD{\BS Y}-\BS X\|_\text F+\|\BS X-\BS W\|_\text F\le
3\epsilon\|\BS X\|_\textnormal F
.\qedhere}
}
\end{proof}

\change{Proposition~\ref{PRP:err} implies that the reconstruction error of our \textsc{Stitch} algorithm is very close to the error of computing Tucker decomposition from scratch via TuckerALS and does not depend on the number $s$ of subtensors to be stitched.}


\section{\Ours{}: Operations}\label{sec:op}
Having described our design of the data structure in Section~\ref{sec:design} and two core algorithms in Section~\ref{sec:algo}, we next introduce how to answer Tucker decomposition range queries in Section~\ref{ssec:op-query} and how to maintain the 
tree after appending a tensor slice in Section~\ref{ssec:op-append}. 

\subsection{Querying over a Time Range}\label{ssec:op-query}
A query over time range $[\Ts,\Te)$ asks to find the Tucker decomposition of the subtensor $\BS X_{[\Ts,\Te)}$. Equipped with the two core algorithms in Section~\ref{sec:algo}, we are ready to present the algorithm for answering the range query. First, we use \ALGref{recall} w.r.t.\ the root of the stream segment tree to find an optimally pruned hit set $\CAL S\subseteq\CAL L\cup\CAL M$. For each partial hit, we use Eq.~\eqref{eq:partial-tucker} to compute its approximate Tucker decomposition $\TLD{\BS Y}_i$ of the subtensor $\BS X_{[\Ts,\Te)\cap[\sigma_{v_i},\tau_{v_i})}$; for each entire hit $v_i$, we retrieve its preprocessed Tucker decomposition $\TLD{\BS Y}_i:=\BS Y_{v_i}$. Same as before, here $\TLD{\BS Y}_1,\dots,\TLD{\BS Y}_s$ are just symbols to refer to the Tucker decomposition products. Finally, we use \ALGref{stitch} to stitch the subtensor decompositions $\TLD{\BS Y}_1,\dots,\TLD{\BS Y}_s$ into the Tucker decomposition $\BS G\times_1\BM U^\Pa1\cdots\times_p\BM U^\Pa p$ of the queried subtensor $\BS X_{[\Ts,\Te)}$. 

The overall procedure can be illustrated using the example in Figure~\ref{fig:illust}. When answering a range query $[\Ts,\Te)=[1,6)$ with $\theta=0.7$, first we use \ALGref{recall} to divide $[1,6)$ into two sub-ranges $[1,4)$ (a partial hit of $[0,4)$) and $[4,6)$ (an entire hit). Since $[1,4)$ is a partial hit of $[0,4)$, then we use Eq.~\eqref{eq:partial-tucker} to approximate the Tucker decomposition of $[1,4)$. Finally, we use \ALGref{stitch} to stitch the decompositions of $\BS X_{[1,4)}$ and $\BS X_{[4,6)}$ into an approximate Tucker decomposition of $[1,6)$.


\begin{proposition}[Time complexity]
\label{PROP:time_complexity_RQ}
Given a query $[\Ts, \Te)$, \Ours{} performs \textsc{Recall} and \textsc{Stitch} operations and takes $O(r^pDs+r^{2p-2}(D+L)+\log T)$ time overall, where the query length $L:=\Te-\Ts$, and the hit set size $s=O(\log L)$.
\end{proposition}

\begin{proof}[\change{Proof sketch}]
\change{For each iteration, there are five computations: (1) the \textsc{Recall} algorithm, (2) the matricization of the temporal mode, (3) the matricization of $p-1=O(1)$ non-temporal modes, (4) Singular value decomposition $p$ times, and (5) the computation for updating core tensor.
Therefore, the overall time complexity is $O(r^{p}Ds + r^{2p-2}D + r^{2p-2}L+\log T)$. 
}
\end{proof}

\ALG{insert}{(\textsc{Insert}): Inserting a leaf node}{
\REQUIRE{current node $v$; current time $T$; tensor slice $\BS X_T$}
\IF{$\sigma_v=T$ and $\tau_v=T+1$}
    \STATE preprocess the Tucker decomposition $\BS Y_v$ of $\BS X_T$
    \STATE$\CAL P\gets\CAL P\setminus\{v\}$
    \STATE$\CAL L\gets\CAL L\cup\{v\}$
\ELSE
    \STATE$\mu\gets\big\lfloor\frac{\sigma_v+\tau_v}2\big\rfloor$
    \IF{$T<\mu$}
        \IF{$v$ does not have a left child $u_1$}
            \STATE create a left child $u_1\in\CAL P$ with $[\sigma_{u_1},\tau_{u_1})\gets[\sigma_v,\mu)$
        \ENDIF
        \STATE$\textsc{Insert}(u_1,T,\BS X_T)$
    \ELSE
        \IF{$v$ does not have a right child $u_2$}
            \STATE create a right child $u_2\in\CAL P$ with $[\sigma_{u_2},\tau_{u_2})\gets[\mu,\tau_v)$
        \ENDIF
        \STATE$\textsc{Insert}(u_2,T,\BS X_T)$
        \IF{$\tau_v=T+1$}
            \STATE$\BS Y_v\gets\textsc{Stitch}(\{\BS Y_{u_1},\BS Y_{u_2}\})$ via \ALGref{stitch}
            \STATE$\CAL P\gets\CAL P\setminus\{v\}$
            \STATE$\CAL M\gets\CAL M\cup\{v\}$
        \ENDIF
    \ENDIF
\ENDIF
}

\ALG{append}{(\textsc{Append}): Appending a tensor slice}{
\REQUIRE{current root node $r$; current time $T$; tensor slice $\BS X_T$}
\ENSURE{the root after appending}
\IF{$T=0$}
    \STATE create a node $r'\in\CAL P$ with $[\sigma_{r'},\tau_{r'})\gets[0,1)$
    \STATE$r\gets r'$
\ELSIF{$T=\tau_r$}
    \STATE create a node $r'\in\CAL P$ with $[\sigma_{r'},\tau_{r'})\gets[0,2\tau_r)$
    \STATE let $r$ be the left child of $r'$
    \STATE$r\gets r'$
\ENDIF
\STATE$\textsc{Insert}(r,T,\BS X_T)$ via \ALGref{insert}
\RETURN$r$
}

\subsection{Appending a Tensor Slice}\label{ssec:op-append}
Appending a tensor slice $\BS X_T$ extends the timespan from $[0,T)$ to $[0,T+1)$. To process this operation, we need to (i) maintain the stream segment tree structure and (ii) update the Tucker decomposition of nodes in the tree. Due to the special structure of the stream segment tree, we cannot maintain a logarithmic height via rotation operations like typical balanced search trees \cite{guibas1978dichromatic}. To address this issue, we leverage the fact that we have only the \emph{appending} operation but no arbitrary insertion operations. Our key idea here is to insert not only a leaf node but also possibly a root node so as to maintain the logarithmic height. In the following, we will first describe the case where we do not need to insert a root node and then discuss the case where we need to insert a root node.

If the root node $r$ is a placeholder node, then its time range $[0,\tau_r)$ includes $\BS X_T$. Thus, we can insert $\BS X_T$ into the tree. The insertion procedure is a recursive algorithm starting from the root $r$. Suppose that we are currently at a node $v$. Let $u_1,u_2$ denote the left and right children of $v$, respectively. If $T<\tau_{u_1}$, then we insert $\BS X_T$ into the subtree rooted at $u_1$. Otherwise, we need to insert $\BS X_T$ into the subtree rooted at $u_2$. If either $u_1$ or $u_2$ does not exist yet, we create that node before insertion. After insertion, we revise the type of the node $v$. If $T=\tau_v-1$, then the range $[\sigma_v,\tau_v)$ has been fully observed, so we let the node $v$ become an intermediate node. The overall insertion procedure is formally presented in \ALGref{insert}.

Meanwhile, if the root node $r$ is already an intermediate node, then its time range $[0,\tau_r)$ has been fully observed. In this case, we create a new placeholder node $r'$ with time range $[0,\tau_{r'}):=[0,2\tau_{r'})$, let node $r$ be the left child of $r$, and let $r'$ be the new root. Since the new root $r'$ is now a placeholder node, then we use \ALGref{insert} to insert $\BS X_T$ into the tree. The overall appending procedure is formally presented in \ALGref{append}.

The overall procedure can be exemplified using Figure~\ref{fig:illust}. Suppose that the current timespan is $[0,8)$ (i.e., the current root is the node $\langle8\rangle$), and that we want to append the slice $\BS X_8$. Since the root node $\langle8\rangle$ is an intermediate node, then we create a new root $\langle16\rangle$ as an intermediate node and let $\langle8\rangle$ be its left child. Then, we insert $\BS X_8$ into the new root. As $T=8$, we need to insert $\BS X_8$ into the right child of $\langle16\rangle$. Since $\langle16\rangle$ does not have a right child yet, we create a new intermediate node $\langle17\rangle$ as its right child and insert $\BS X_8$ into $\langle17\rangle$. We repeat this procedure until it reaches the leaf node $\langle20\rangle$. We preprocess the Tucker decomposition of $\BS X_{[8,9)}$ via Tucker-ALS and store it at node $\langle20\rangle$.

\begin{theorem}[Logarithmic height]\label{THM:alg-log-height}
If we use the algorithm above for timespan $[0,T)$, then it can construct a stream segment tree of height $\lceil\log_2T\rceil+1$.
\end{theorem}
\begin{proof}[\change{Proof sketch}]
\change{We can use an induction on $T$ to show that at time $T$, the number of new leaf nodes that can be inserted into the sub-tree rooted at each node $v$ is $\max\{\tau_v-T,0\}$. This implies that the number of leaf nodes in a stream segment tree of height $h$ is greater than the number of leaf nodes of a perfect binary tree of height $h-1$ and at most that of a perfect binary tree of height $h$. It follows that the height of a stream segment tree is $\lceil\log_2T\rceil+1$. 
}
\end{proof}

\begin{proposition}[Time complexity]
\label{PROP:time_complexity_appending}
When appending a tensor slice $\BS X_T$, the amortized and worst-case time complexities are $O(rD^{p-1} + r^{2p-2}(D +\log T))$ and $O(rD^{p-1} + r^{2p-2}(D\log T + T))$, respectively.
\end{proposition}
\begin{proof}[\change{Proof sketch}]
\change{
We analyze the worst-case and amortized complexities of a stream segment tree.
The amortized complexity is the time complexity of creating a stream segment tree divided by $T$.
We preprocess $T$ tensor slices of leaf nodes with the complexity $O(rD^{p-1}T)$.
For intermediate nodes, we perform $O(T)$ stitch operations with the complexity $O(r^{2p-2}DT + r^{2p-2}T\log T)$.
Hence, the amortized complexity is $O(rD^{p-1} + r^{2p-2}D+r^{2p-2}\log T)$.
Appending a tensor slice $\BS X_T$ requires three computations: (1) performing Tucker-ALS of the tensor slice, (2) updating the non-temporal factor matrices, and (3) updating the temporal factor matrices.
Hence, the worst-case complexity of appending a tensor slice $\BS X_T$ is $O(rD^{p-1} + r^{2p-2}D\log T + r^{2p-2}T)$ which is the sum of the complexities of the three computations.
}
\end{proof}

\section{Implementation}



Since the bottleneck of Tucker decomposition is the tensor numerical operations, we speed up these computations using a graphical processing unit (GPU). Though originally designed to accelerate computer graphics and image processing, modern GPUs are powerful in parallelizing dense numerical operations in general scientific computing. To develop a prototype of our \Ours{} that is compatible across various platforms, we choose the PyTorch \cite{pytorch} CUDA \cite{cuda} library to set up and run GPU operations. Although Python execution is relatively slow compared with many other programming languages, it does not affect the overall performance much because the tensor operations are typically much more expensive than Python execution. 



\begin{table}[t!]
\caption{Summary of real-world tensor time series datasets. 
	}
\centering\label{tab:Description}
\resizebox{\linewidth}{!}{
\begin{tabular}{cccc}
\toprule
\textbf{Dataset}&\textbf{\#Entries} & \textbf{Shape} & \textbf{Modes} \\
\midrule	
Air Quality&47M &  $21000\times 376\times 6$&(time, location, air pollutant) \\
Traffic  &212M&  $2033\times 1084\times 96$&(time, frequency, sensor) \\
US Stock &739M&  $2000\times 4347\times 85$&(time, company, stock feature) \\
KR Stock &875M&  $3000\times 3432\times 85$&(time, company, stock feature) \\
\bottomrule
\end{tabular}}
\end{table}
\begin{figure*}[t]
\centering
\setcounter{subfigure}{0}
\subfloat[Air Quality]{\includegraphics[width=0.211\textwidth]{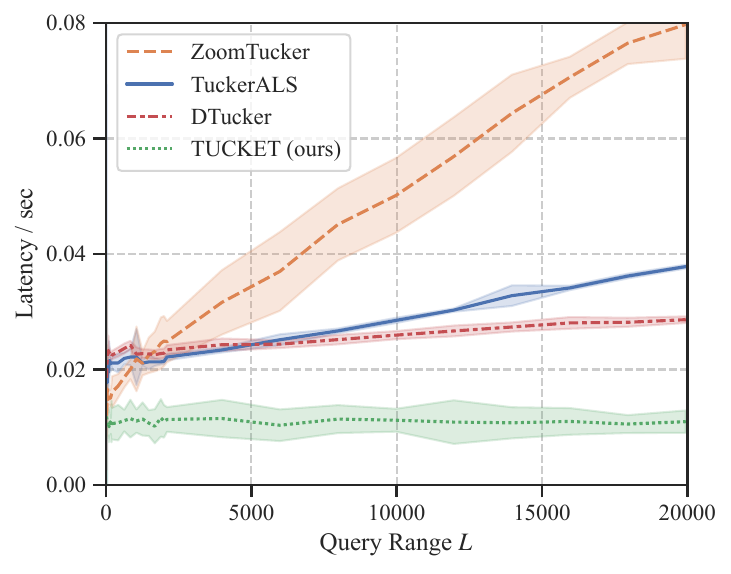}\label{fig:exp-lat-air}}
\subfloat[Traffic]{\includegraphics[width=0.2\textwidth]{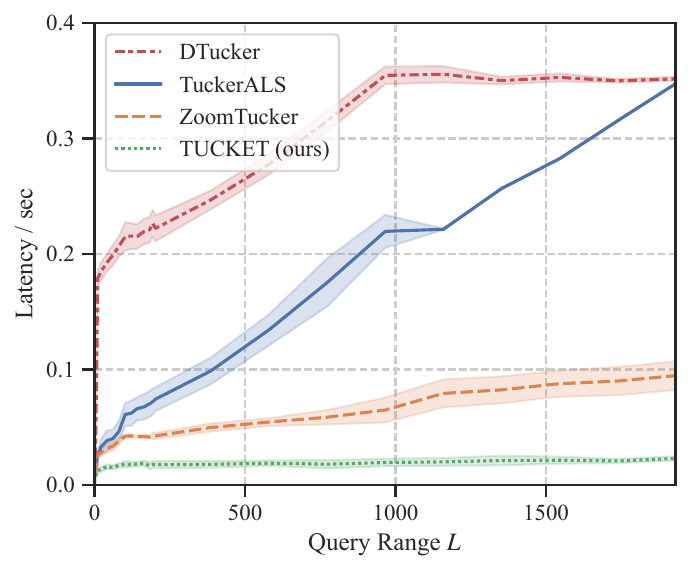}\label{fig:exp-lat-traffic}}
\subfloat[US Stock]{\includegraphics[width=0.2\textwidth]{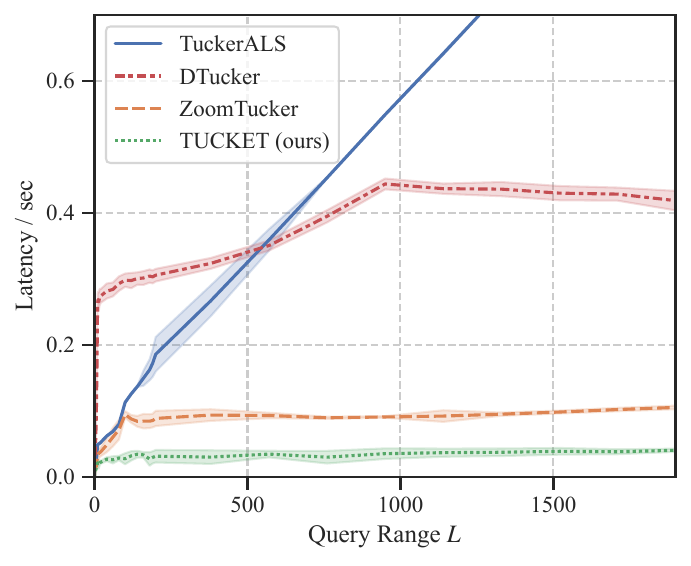}\label{fig:exp-lat-us}}
\subfloat[KR Stock]{\includegraphics[width=0.2\textwidth]{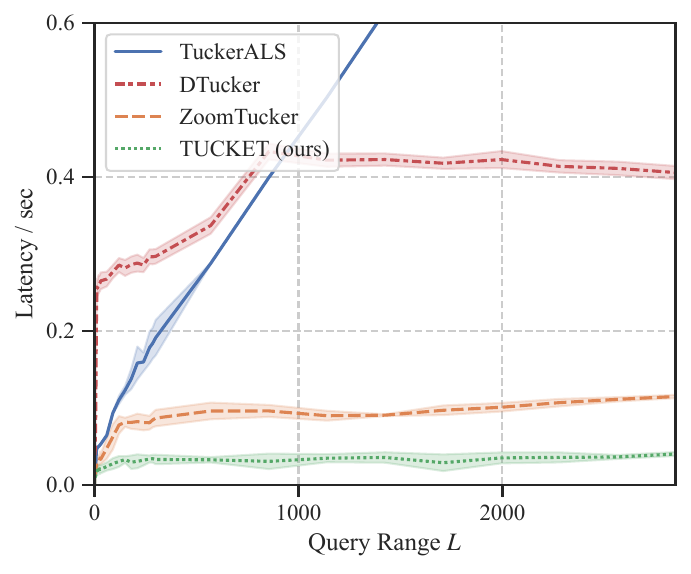}\label{fig:exp-lat-kr}}
\caption{
Comparison in the latency of range queries. 
Our \Ours{} (green dotted line) consistently achieves the lowest latency for all query ranges while the performance of other methods varies drastically across datasets.
}
\label{fig:exp-lat}
\end{figure*}




\begin{figure*}[t]
\centering
\captionsetup[subfigure]{justification=centering}
\setcounter{subfigure}{0}
\subfloat[\change{Latency v.s.\ $p$}]{\includegraphics[width=0.22\textwidth]{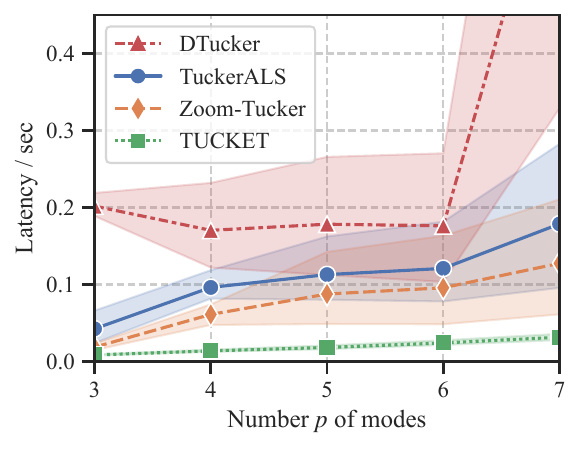}\label{fig:exp-l-p}}
\subfloat[\change{Latency v.s.\ $T$}]{\includegraphics[width=0.2425\textwidth]{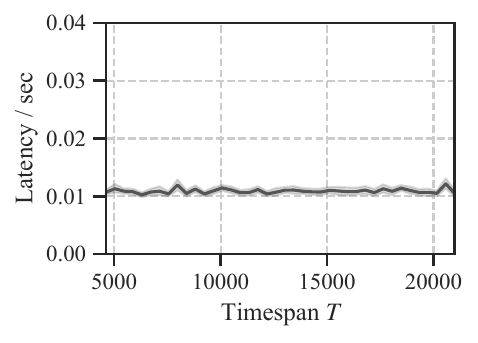}\label{fig:exp-lt}}
\subfloat[Latency of \textsc{Append}]{\includegraphics[width=0.2475\textwidth]{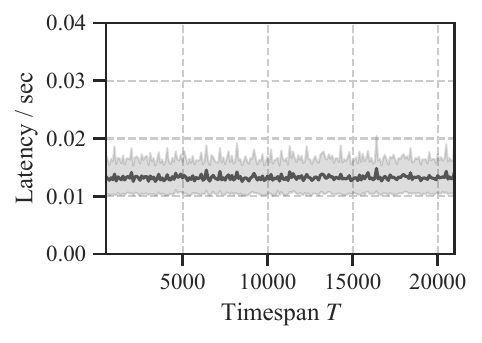}\label{fig:exp-cum-time}}
\subfloat[Cumulative space]{\includegraphics[width=0.24\textwidth]{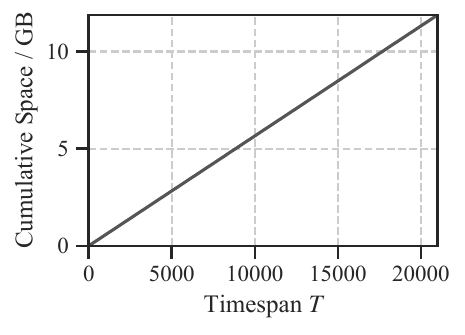}\label{fig:exp-cum-space}}
\caption{
Scalability tests
}
\label{fig:exp-cumulative}
\end{figure*}

\begin{figure*}[t]
\centering
\setcounter{subfigure}{0}
\subfloat[Air Quality]{\includegraphics[width=0.2\textwidth]{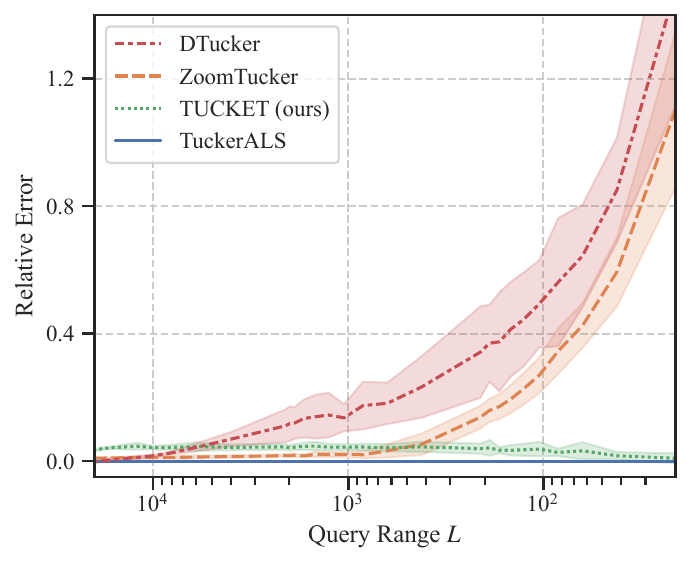}\label{fig:exp-err-air}}
\subfloat[Traffic]{\includegraphics[width=0.2\textwidth]{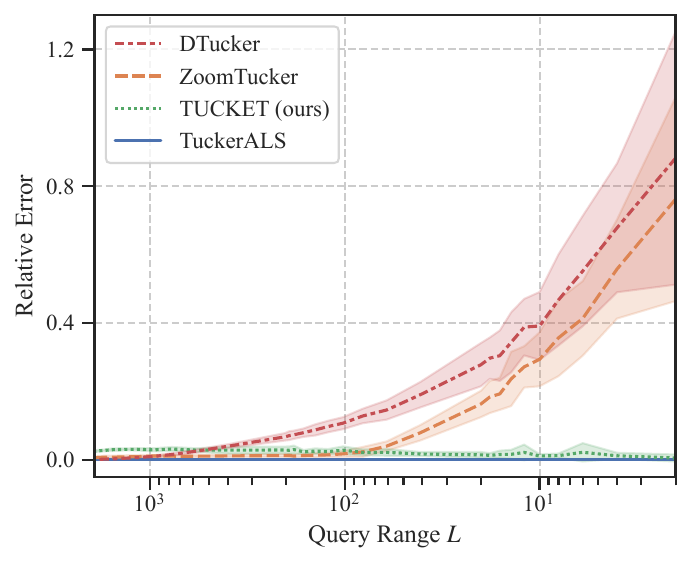}\label{fig:exp-err-traffic}}
\subfloat[US Stock]{\includegraphics[width=0.2\textwidth]{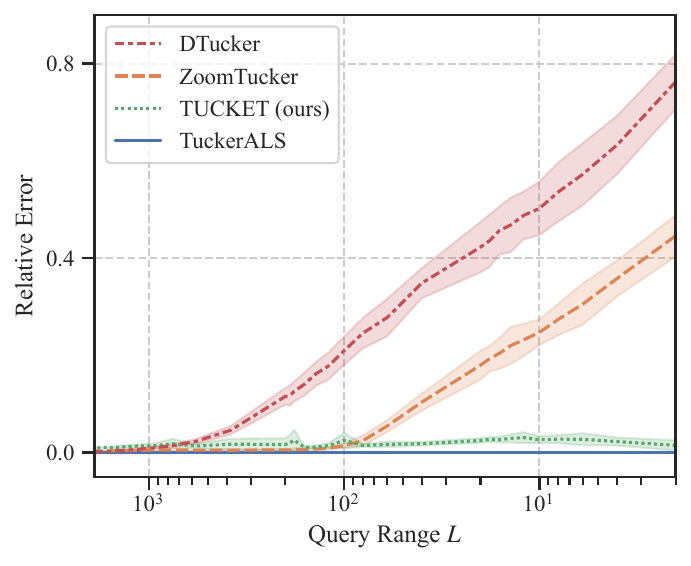}\label{fig:exp-err-us}}
\subfloat[KR Stock]{\includegraphics[width=0.2\textwidth]{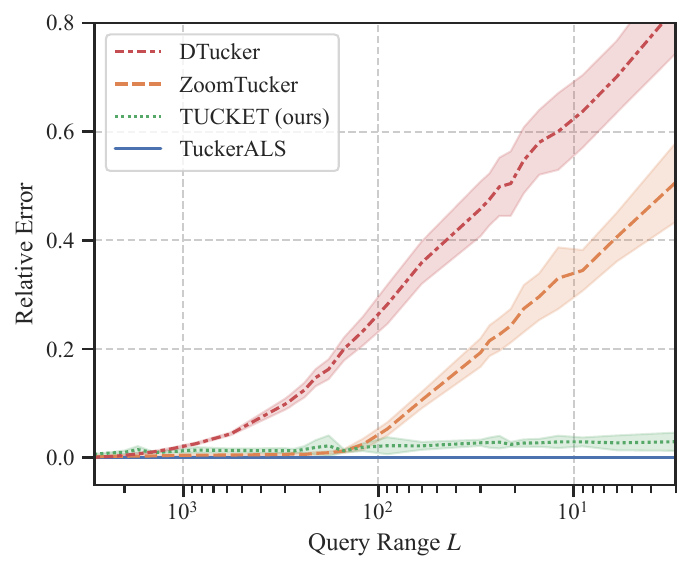}\label{fig:exp-err-kr}}
\caption{
Comparison in the relative error of range queries.
In contrast to D-Tucker and Zoom-Tucker, our \Ours{} (green dotted line) consistently achieves comparable 
error with that of Tucker-ALS for all query ranges.
}
\label{fig:exp-err}
\end{figure*}

\begin{figure}[t]
\captionsetup[subfigure]{justification=centering}
\centering
\setcounter{subfigure}{0}
\subfloat[\change{Cumulative latency}]{\includegraphics[width=0.43\linewidth]{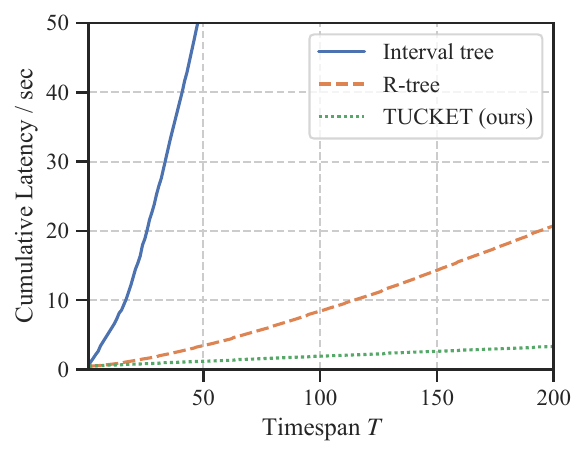}\label{fig:trees-dur}}
\subfloat[\change{Cumulative stitches}]{\includegraphics[width=0.45\linewidth]{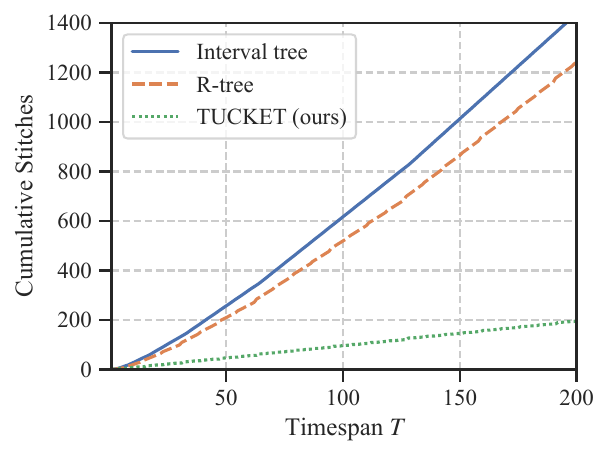}\label{fig:trees-stit}}
\caption{
\change{Comparison with the interval tree and the R-tree in terms of the \textsc{Append} operation.}
}
\label{fig:trees}
\end{figure}

\begin{figure}[t]
\captionsetup[subfigure]{justification=centering}
\centering
\setcounter{subfigure}{0}
\subfloat[Latency v.s.\ \\ stitching algorithm]{\includegraphics[width=0.5\linewidth]{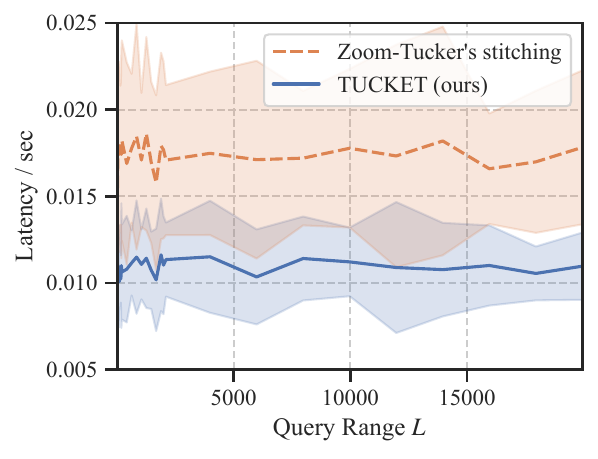}\label{fig:exp-stitch-latency}}
\subfloat[Relative error v.s.\ \\ stitching algorithm]{\includegraphics[width=0.48\linewidth]{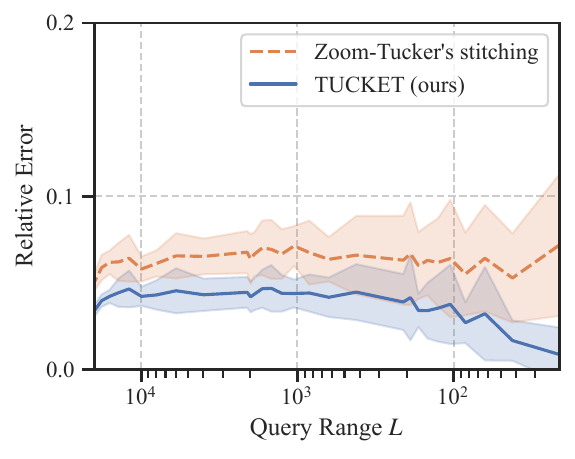}\label{fig:exp-stitch-error}}
\caption{
Comparison with Zoom-Tucker's stitching
}
\label{fig:exp-stitch}
\end{figure}

\begin{figure}[t]
\captionsetup[subfigure]{justification=centering}
\centering
\setcounter{subfigure}{0}
\subfloat[Latency v.s.\ \\ pruning threshold $\theta$]{\includegraphics[width=0.515\linewidth]{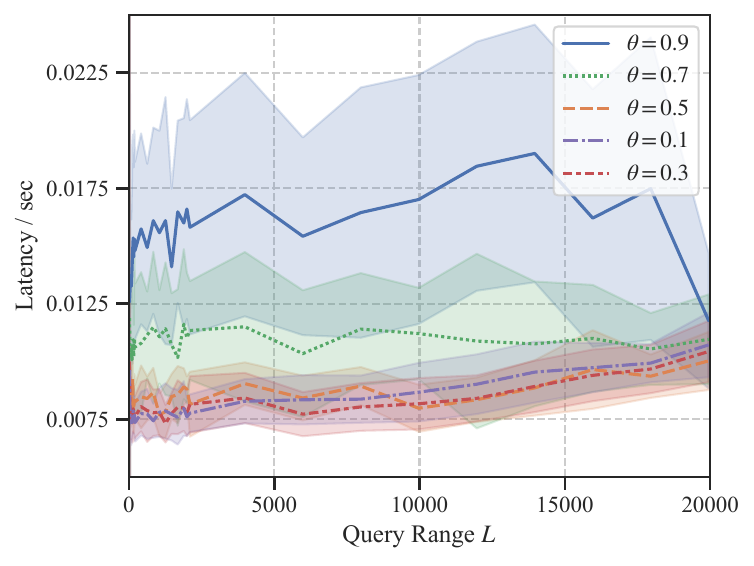}\label{fig:exp-prune-latency}}
\subfloat[Relative error v.s.\ \\ pruning threshold $\theta$]{\includegraphics[width=0.47\linewidth]{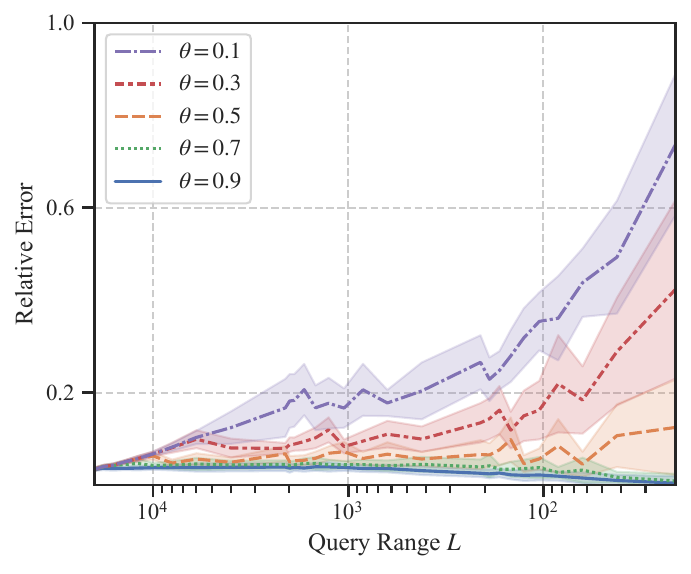}\label{fig:exp-prune-error}}
\caption{
Latency and relative error v.s.\ pruning threshold
}
\label{fig:exp-prune}
\end{figure}



\section{Experimental Evaluation}
In this section, we evaluate our \Ours{} by comparing it with state-of-the-art methods on four large-scale real-world tensor time series datasets. We summarize our evaluation results from our experiments as follows:
\begin{enumerate}[noitemsep,topsep=0pt]
\renewcommand{\labelenumi}{(\roman{enumi})}
\item \Ours{} consistently achieves the lowest latency over all query ranges on real-world tensor time series data \textbf{(Section~\ref{subsec:exp-efficiency})}.
\item We empirically demonstrate that \Ours{} constructs the whole tree in nearly linear time and consumes nearly linear space in total \textbf{(Section~\ref{subsec:exp-efficiency})}.
\item The reconstruction error of \Ours{} is much smaller than D-Tucker and Zoom-Tucker and is comparable with the brute-force method Tucker-ALS \textbf{(Section~\ref{subsec:exp-accuracy})}. 
\item Our new stitching algorithm is more GPU-parallelizable and more numerically stable than that of Zoom-Tucker \textbf{(Section~\ref{subsec:exp-sensitivity})}.
\item The pruning threshold $\theta=0.7$ can achieve both high efficiency and accuracy in our experiments 
\textbf{(Section~\ref{subsec:exp-sensitivity})}.
\end{enumerate}


\subsection{Experimental Settings}\label{ssec:experiment-setting}


\Par{Datasets}
We use four large-scale real-world tensor time series datasets, which are summarized in Table~\ref{tab:Description}.
\textbf{(D1)} Air Quality data is represented as $3$-way tensor time series (time, location, air pollutant). It is collected from the Air Korea\footnote{\url{https://www.airkorea.or.kr/web/}} website.
\textbf{(D2)} Traffic data\footnote{\url{https://github.com/florinsch/BigTrafficData}}~\cite{schimbinschi2015traffic} is $3$-way tensor time series (time, frequency, sensor) representing a collection of traffic volume measurements around Melbourne. 
\textbf{(D3 \& D4)} We use daily stock features (e.g., prices, volumes, and technical indicators) on the U.S. and Korean stock markets, respectively, to build $3$-way tensor time series (time, company, stock feature).
\textbf{(D5)} To evaluate the scalability w.r.t.\ the number $p$ of modes, we generate synthetic tensors with the following sizes: (1) $1000\times 1000\times 1000$, (2) $1000\times 100 \times 100 \times 100$, (3) $1000\times 100 \times 100 \times 10 \times 10$, (4) $1000\times 100 \times 10 \times 10 \times 10 \times 10$, (5) $1000\times 10 \times 10 \times 10 \times 10 \times 10 \times 10$.

\Par{Evaluation metrics} 
Regarding efficiency, since no baseline methods support batch operations, we do not report the throughput. Instead, we report the \emph{latency} (in seconds) of each operation.

Regarding accuracy, for each range query $[\Ts,\Te)$, we report the \emph{relative error} between each method $A$ and Tucker-ALS:
\begin{align}
    \frac{\|{\BS X}_{[\Ts,\Te)} - {\BS Y}_{A} \|_\text F}{\|{\BS X}_{[\Ts,\Te)} - {\BS Y}_\text{Tucker-ALS} \|_\text F} - 1,
\end{align}
where ${\BS Y}_\text{Tucker-ALS}$ and ${\BS Y}_{A}$ denote the reconstructed tensors by Tucker-ALS and the method $A$, respectively. 

\Par{Platform} 
We conduct all experiments in the Docker environment on an Ubuntu 20.04.6 LTS cloud server with an Intel Xeon CPU @ 2.00 GHz and an NVIDIA P100 16GB GPU.

\Par{Evaluation framework}
When comparing the performance of \Ours{} with baselines, since no baseline methods support range queries and stream updates simultaneously, we first construct the data structures of all methods and then run range queries. 
The appending operation is evaluated separately in Section~\ref{subsec:exp-efficiency}.

\Par{Hyperparameters}
For alternating least squares, we set the maximum number of iterations to 20 and the tolerance to 0.01. 
We set a target size to $10$, except that we set the target size to 5 when the size of a non-temporal mode is smaller than $10$.
For \Ours{}, we use the pruning threshold $\theta=0.7$ for all datasets.

\subsection{Baseline Methods}
To evaluate the effectiveness of our proposed \Ours{}, we compare it with state-of-the-art methods for Tucker decomposition.
The baselines are described below. 

\begin{itemize}[noitemsep,topsep=0pt]
\item \textbf{Tucker-ALS} \cite{tuckerals} utilizes alternating least squares optimization to compute Tucker decomposition, thus achieving the best reconstruction error. However, in our setting, 
Tucker-ALS is essentially 
a brute-force algorithm, unable to support range queries or streaming updates efficiently.
\item \textbf{D-Tucker} \cite{dtucker} decomposes compressed matrices sliced from the input tensor and further updates the factors and cores iteratively, which enables a fast and memory-efficient decomposition of large and dense tensors. Notably, the iteration phase of D-Tucker facilitates its seamless adaptation to tasks involving stream updates. However, it does not support 
range queries. For a fair comparison, we extract the sub-tensor corresponding to the range query from the preprocessed slices of D-Tucker.
\item \textbf{Zoom-Tucker} \cite{zoomtucker} supports Tucker decomposition range queries via block-wise preprocessing and by merging block results during the query pahse. 
While Zoom-Tucker demonstrates efficient performance on range queries, it faces limitations in supporting stream updates. For a fair comparison, we use block size $\frac T{2\lceil\log_2T\rceil}$ in Zoom-Tucker, which ensures that the maximum number of blocks is comparable with the maximum size of the hit set of \Ours{}. Besides that, since we have enhanced the stitching algorithm of subtensor decompositions in Section~\ref{ssec:core-stitch}, we also use our stitching algorithm in Zoom-Tucker for a fair comparison.
\end{itemize}


\subsection{Efficiency \& Scalability Tests}
\label{subsec:exp-efficiency}
In this subsection, we evaluate the time and space efficiencies of \Ours{} in range query answering and appending tensor slices.

\Par{Efficiency of range query answering w.r.t.\ query length $L$}
As shown in Figure~\ref{fig:exp-lat}, \Ours{} consistently achieves the lowest latency compared to all baseline methods, with its latency remaining almost stable regardless of the query range. In contrast, the latency of baselines increases dramatically as the query range expands.
On the Air Quality dataset, \Ours{} exhibits an average latency of 0.011 seconds on a GPU, significantly smaller than that of all other baselines. In Traffic, US Stock, and KR Stock datasets, although Zoom-Tucker and \Ours{} demonstrate similar trends, our \Ours{} outperforms Zoom-Tucker by a considerable margin. 

\Par{\change{Scalability w.r.t.\ number $p$ of modes}}
\change{We measure latencies of \Ours{} and baselines by varying the number $p$ of modes on synthetic tensor time series under query lengths $98$, $192$, and $384$ and 
target rank $r=5$.
As shown in Figure~\ref{fig:exp-l-p}, \Ours{} is still the most efficient method under a higher number $p$ of modes. \Ours{} consistently achieves the lowest latency compared to all baselines across all query lengths and number of modes. \Ours{} achieves $5.9$ times lower latency than the second-fastest method, Zoom-Tucker, when the number of modes $p$ is $7$ and the query length is $384$. This highlights the superiority of our recall and stitching algorithms 
in terms of the scalability w.r.t.\ the number $p$ of modes.}

\Par{\change{Efficiency of query answering w.r.t.\ timespan $T$}} \change{We keep the query length $L$ the same while appending new slices between queries to increase $T$. The results of latency v.s.\ timespan $T$ on Air Quality are shown in Figure~\ref{fig:exp-lt}. We can see that the latency of range queries almost has no notable change. This validates our time complexity where the dominant term $O(r^pD\log L)$ scales with only $L$ and does not explicitly depend on $T$.
}

\Par{Efficiency of appending tensor slices}
We plot the latency of \textsc{Append} on the Air Quality dataset v.s.\ the timespan $T$ in Figure~\ref{fig:exp-cum-time}. The results validate the amortized time complexity $O(rD^{p-1}+r^{2p-2}(D+\log T))$ of \textsc{Append}. Notably, Figure~\ref{fig:exp-cum-time} shows that the time complexity is nearly constant w.r.t.\ $T$. This is because $D$ is typically much greater than $\log T$ as long as $T$ is not too large. 

\Par{Space consumption of our \Ours{}}
We plot the cumulative space usage on the Air Quality dataset v.s.\ the timespan $T$ in Figure~\ref{fig:exp-cum-space}. The results validate the space complexity $O((rD+r^p)T+rT\log T)$ of our \Ours{}. Notably, Figure~\ref{fig:exp-cum-space} shows that the space complexity is nearly linear w.r.t.\ $T$. This is because $D$ is typically much greater than $\log T$ as long as $T$ is not too large. 

\subsection{Accuracy of Range Query Answering}
\label{subsec:exp-accuracy}
We measure relative errors with respect to time range queries.
Figure~\ref{fig:exp-err} shows the results.
\Ours{} consistently has comparable errors to Tucker-ALS which performs Tucker decomposition from scratch.
As $T_e - T_s$ decreases, \Ours{} and Tucker-ALS have little variation in errors, while the errors of D-Tucker and Zoom-Tucker increase dramatically.
This is because \Ours{} effectively preserves information for short time ranges using the proposed stream segment tree in the preprocessing phase, whereas D-Tucker and Zoom-Tucker compromise the accuracy for short time ranges in the preprocessing phase.

\subsection{Ablation Studies}
\label{subsec:exp-sensitivity}
To further understand 
\Ours{}, we conduct the following ablation studies: (i) comparing our stream segment tree with other data structures, (ii) comparing our new stitching algorithm with that of Zoom-Tucker, and (iii) analyzing the effect of pruning threshold $\theta$. 

\Par{\change{Comparison of data structures}} \change{
We compare our stream segment tree with the interval tree \cite{preparata2012computational} (using AVL balancing \cite{adel1962algorithm}) and the (1-dimensional) R-tree \cite{guttman1984r} (using B-tree balancing \cite{bayer1970organization}) in terms of the \textsc{Append} operation. 
We report in Figure~\ref{fig:trees} the cumulative latency and the cumulative number of \textsc{Stitch} operations in \textsc{Append} because the \textsc{Stitch} operation is the bottleneck during \text{Append}. We see that our \Ours{} achieves 83.5 times and 3.4 times speedup over the interval tree and the R-tree, respectively; they need $O(\log T)$ \textsc{Stitch} operations per \text{Append} because every node on the path from the inserted node to the root needs a \textsc{Stitch}. In contrast, our proposed \emph{stream segment tree} needs only 1 \textsc{Stitch} (amortized) per \text{Append} because our placeholder nodes need no \textsc{Stitch}. 
}

\Par{Comparison of stitching algorithms}
We compare our stitching algorithm with Zoom-Tucker's stitching algorithm in terms of latency and relative error by replacing our \textsc{Stitch} (\ALGref{stitch}) with Zoom-Tucker's stitching algorithm. 
In Figure~\ref{fig:exp-stitch-latency}, our stitching algorithm achieves lower latency than the stitching algorithm of Zoom-Tucker.
This result implies that our stitching algorithm is more GPU-parallelizable than Zoom-Tucker's stitching.
Figure~\ref{fig:exp-stitch-error} shows that our stitching algorithm has lower relative errors than Zoom-Tucker's stitching. 
The error gap widens 
as the query range decreases since Zoom-Tucker's stitching needs to compute the inverse for rank-deficient matrices in short query ranges.

\Par{Effect of the pruning threshold $\theta$}
We test the model efficiency and the relative error of \Ours{} with respect to the pruning threshold. We conduct the experiment on the Air Quality dataset.
Figure~\ref{fig:exp-prune-error} shows the relative error of \Ours{} with respect to the query range. As we can see, as the $\theta$ value increases, the relative error consistently decreases for all query ranges. Meanwhile, the difference between $\theta=0.7$ and $\theta=0.9$ is tiny. 
When comparing $\theta=0.7$ and $\theta=0.9$, $\theta=0.7$ is better 
since it has lower latency than $\theta=0.9$ over all the query ranges as shown in Figure~\ref{fig:exp-prune-latency}.
Therefore, $\theta=0.7$ is the best choice for balancing efficiency and accuracy as it allows \Ours{} to avoid pruning-induced accuracy loss and handle a similar number of blocks.

\subsection{\change{Case Study}}
\label{subsec:exp-case-study}
\change{
To demonstrate the application of \Ours{}, we present a case study with Air Quality data. 
We run \Ours{} for three time ranges (March 2015, March 2016, and March 2017) to obtain the location factor matrices $\BM U^\Pa2 \in \mathbb{R}^{376\times r}$ of each time range. Here, the $i$-th row vector of $\BM U^\Pa2$ represents the air pollution patterns of the $i$-th location. Then, we perform K-means clustering w.r.t.\ the row vectors of $\BM U^\Pa2$ to find 5 clusters of the locations. }

\change{Clustering results are shown in Figure~\ref{fig:exp-case-air}. We can see that 
some regions consistently exhibit similar clustering patterns across all years while some other regions have varying clustering patterns depending on the year.
Regions \text{(A)} and \text{(D)} had similar patterns in March 2015 and 2016, but their patterns changed 
in March 2017.
Meanwhile, regions \text{(B)}, \text{(C)}, and \text{(E)} had similar clustering patterns across all years.
Region \text{(F)} has slightly different clustering patterns across the years.
Air pollution patterns can be attributed to changes in weather conditions, the occurrence of yellow dust and fine dust, industrial activity, and traffic volume.
\Ours{} enables researchers and practitioners to explore diverse time ranges on Air Quality data efficiently and accurately. 
}
\section{Related Work}

\Par{Tensor decomposition}
Tensor decompositions 
have been widely used for analyzing real-world tensors.
Due to their large sizes, developing efficient and scalable CP methods~\cite{li20162pcp,huang2016bicp,cheng2016spals,battaglino2018practical,ballard2018parallel,zhang2023scalable} and Tucker methods~\cite{tsourakakis2010mach,dtucker,tuckerts,ma2021fast,jang2023static} have attracted considerable interest.
The vast majority of these works focus on decomposing the entire tensor from scratch and thus cannot efficiently answer range queries. 
Although Zoom-Tucker~\cite{zoomtucker} supports efficient range queries, it has an unwilling tradeoff between accuracy and efficiency due to a fixed block size. 
In addition to the above methods for dense tensors, there are plenty of tensor decomposition methods for sparse tensors where only a few entries are nonzeros.
Many works~\cite{papalexakis2012parcube,kaya2015scalable,kaya2016high,yang2017lftf,smith2017accelerating,oh2017s,oh2018scalable,oh2019high,li2020sgd,shi2023parallel} have developed scalable tensor decomposition for sparse tensors in parallel and distributed systems.
Numerous tensor decomposition methods~\cite{liu2019costco,wu2019neural,chen2020neural,tillinghast2020probabilistic,fan2021multi} utilize neural networks for predicting unobserved entries of sparse tensors.
However, they do not support range queries either. 


\Par{Time series databases}
Time series databases utilize various data structures to handle time-series data.
Many time series databases, including KairosDB~\cite{kairosdb}, Apache IoTDB~\cite{wang2023apache}, and LittleTable~\cite{rhea2017littletable}, are designed based on log-structured merge tree (LSM-tree)~\cite{o1996log} for managing time-series data.
Other time series databases, including InfluxDB~\cite{influxdata}, BTrDB~\cite{andersen2016btrdb}, and EdgeDB~\cite{yang2019edgedb}, utilize their own tree structures to manage massive time series data. 
However, none of these works considers tensor time series or supports tensor decomposition range queries, which is harder 
than the simple queries supported by existing time series databases.

\section{Conclusion \& Future Work}
In this paper, we have proposed a tensor time series data structure called \Ours{} that can efficiently and accurately support range queries of Tucker decomposition and stream updates to the tensor time series. To the best of our knowledge, our \Ours{} is a first-of-its-kind method that creatively generalizes the segment tree to the Tucker decomposition range query problem with stream updates. We provide (i) fine-grained theoretical guarantees and (ii) time and space complexities for our proposed method. We also experimentally show that \Ours{} consistently achieves the best efficiency and accuracy in time range query answering.

Future work includes extending \Ours{} to sparse tensors and to other tensor decompositions such as CANDECOMP/PARAFAC and PARAFAC2 decompositions, 
and supporting multi-mode range queries via nested segment trees \cite{vaishnavi1982computing}. 

\begin{acks}
This work was supported by NSF (2316233), 
NIFA (2020-67021-32799), 
and IBM--Illinois Discovery Accelerator Institute. 
The content of the information in this document does not necessarily reflect the position or the policy of the Government, and no official endorsement should be inferred.  The U.S. Government is authorized to reproduce and distribute reprints for Government purposes notwithstanding any copyright notation here on.
Jun-Gi Jang was supported by Basic Science Research Program through the National Research Foundation of Korea (NRF) funded by the Ministry of Education (RS-2023-00238596).
\end{acks}

\bibliographystyle{ACM-Reference-Format}
\bibliography{80-ref}

\appendix\clearpage\newpage\section{Preliminaries on Basic Tensor Operations}\label{app:prelim}
A \emph{$p$-way tensor} can be viewed as a $p$-dimensional array. Each dimension of a tensor is called a \emph{mode}. To distinguish tensors from matrices and scalars, we use a bold calligraphic font for tensors, a bold italic font for matrices, and non-bold fonts for scalars. The indices of vectors, matrices, and tensors start from $0$. 

The \emph{vectorization} of a tensor $\BS X\in\BB R^{D_1\times\dots\times D_p}$ is a column vector 
where each entry $\SCR X_{i_1,\dots,i_p}$ goes to $(\VEC(\BS X))_j$ at 
\begin{equation}
j=\sum_{n=1}^pi_n\prod_{m=n+1}^{p}D_m.
\end{equation}
The \emph{mode-$n$ matricization} of a tensor $\BS X\in\BB R^{D_1\times\dots\times D_n\times\dots\times D_p}$ is a matrix $\Mat{n}(\BS X)\in\BB R^{D_n\times(D_1\cdots D_{n-1}D_{n+1}\cdots D_p)}$ defined by stacking the mode-$n$ slices of $\BS X$ into a matrix: 
\begin{equation}
\Mat{n}(\BS X):=\MAT{(\VEC(\BS X_{:,\dots,:,0,:,\dots,:}))^\Tp\\\vdots\\(\VEC(\BS X_{:,\dots,:,D_{n}-1,:,\dots,:}))^\Tp}.
\end{equation}
The \emph{mode-$n$ product} of a tensor $\BS G\in\BB R^{r_1\times\dots\times r_n\times\dots\times r_p}$ with a matrix $\BM U\in\BB R^{D\times r_n}$ is a tensor $\BS G\times_n\BM U\in\BB R^{r_1\times\cdots\times r_{n-1}\times D\times r_{n+1}\times\cdots\times r_p}$ defined by
\begin{equation}
\Mat{n}(\BS G\times_n\BM U):=\BM U\cdot\Mat{n}(\BS G),
\end{equation}
where $\cdot$ denotes the matrix multiplication.
Besides that, the \emph{Frobenius norm} $\|\cdot\|_\text F$ of a tensor $\BS X\in\BB R^{D_1\times\dots\times D_p}$ is the square root of the sum of the squares of all its entries:
\begin{equation}
\|\BS X\|_\text F:=\sqrt{\sum_{i_1=0}^{D_1-1}\cdots\sum_{i_p=0}^{D_p-1}\SCR X_{i_1,\dots,i_p}^2}.
\end{equation}
We refer readers to \cite{kolda2009tensor} for further details on tensor operations.\section{Proofs}\label{app:proof}

\subsection{Proof of Lemma~\ref{LEM:hit-height}} \label{PRF:hit-height}

Given a stream segment tree of height $h$ and any range query $\left[ T_s, T_e \right)$, let $\left[ \sigma_v, \tau_v \right)$ be the time range of the root node $v$. Then, when we consider how to find a hit set of range query, there are only three different cases:
\begin{enumerate}
    \item The range query $\left[ T_s, T_e\right)$ is a prefix of the time range $\left[ \sigma_v, \tau_v \right)$, \ie, $T_s = \sigma_v$.
    \item The range query $\left[ T_s, T_e \right)$ is a suffix of the time range $\left[ \sigma_v, \tau_v \right)$, \ie, $T_e = \tau_v$.
    \item The range query $\left[ T_s, T_e \right)$ is neither a prefix nor a suffix of the time range $\left[ \sigma_v, \tau_v \right)$, \ie, $T_s \neq \sigma_v \text{ and } T_e \neq \tau_v$.
\end{enumerate}

For case (1), let $f_1(h)$ be the size of a hit set on a stream segment tree of height $h$. It is evident that $f_1(1) = 1$ since the root node constitutes an entire hit under this circumstance. When $h > 1$, we need to consider whether the range query $\left[ T_s, T_e\right)$ can be entirely contained within the range of the left children. Let $u$ represent the left child of the node $v$. If the query range can be entirely entained, \ie, $T_e \leq \tau_u$, then we should recursively search the hit set within the sub-tree with $u$ as the new root node, implying $f_1(h) \le f_1(h-1)$. Otherwise, we simply put the left child $u$ into the hit set and recursively search within the right sub-tree, which implies $f_1(h) \le f_1(h-1) + 1$. To sum up, we can easily have the following expression
\begin{equation}
    f_1(h) \leq f_1(h-1) + 1
\end{equation}
Since we have the base case that $f_1(1) = 1$, it is easy to prove that
\begin{equation}
    f_1(h) \leq h.
\end{equation}
Similarly, we can also prove that $f_2(h) \leq h$, where $f_2(h)$ represents the size of a hit set for case (2).

Now consider case (3) where $f_3(h)$ is denoted as the size of the hit set. Suppose $u$ is the left child of the node $v$ with the time range of $\left[ \sigma_u, \tau_u\right)$, if the range query spans the time range of both the left child and the right child, \ie, $T_s < \tau_u < T_e$, then we are supposed to search for the hit set in both the left sub-tree and the right sub-tree. Please note that in this situation, the new range query will be either the prefix or the suffix for those sub-trees, \ie, $f_3(h) \le f_1(h-1) + f_2(h-1)$. Conversely, if the range query are limited within the time range of one sub-tree, then the hit set should be searched within the sub-tree, \ie, $f_3(h) \le f_3(h -1)$. To summarize, we have the following expression
\begin{equation}
    f_3(h) \leq f_1(h-1) + f_2(h-1) \leq 2(h-1).
\end{equation}

It follows from the three cases that there exists a hit set of size
\AL{\le\max\{f_1(h),f_2(h),f_3(h)\}=\max\{2(h-1),h\}=\max\{2(h-1),1\}.}

\subsection{Proof of Theorem~\ref{THM:height-T}}

Using the algorithm in Section~\ref{ssec:op-append} to append the tensor slices $\BS X_0,\dots,\BS X_{T-1}$ one by one, we can build a stream segment tree over $[0,T)$. By Theorem~\ref{THM:alg-log-height}, this stream segment tree has height $\lceil\log_2T\rceil+1$.

\subsection{Proof of Proposition~\ref{PROP:space_complexity}}
In a stream segment tree over the range $[0,T)$, there are $O(T)$ nodes each of which has $p-1$ non temporal factor matrices of the size $O(rD)$ and a core tensor of the size $O(r^p)$.
In addition, at each level of the tree, the sum of the sizes of temporal factor matrices is $O(rT)$.
Therefore, the space complexity of the stream segment tree is $O((prD+r^p)T + rT\log T)$.

\subsection{Proof of Theorem~\ref{THM:prune-opt}}
\textbf{Optimal hit sets.} We aim to prove that \ALGref{recall} successfully finds an optimal solution for Eq.~\eqref{eq:prune-obj}. Mathematically, we want to show that,  if there exists a hit set that satisfies the conditions specified in Eq. \eqref{eq:prune-c1} and \eqref{eq:prune-c2}, then the size of this hit set must be greater than or equal to the size of the hit set generated by \ALGref{recall}.

First, due to the top-down nature of \ALGref{recall}, for any node $v$ in the hit set generated by \ALGref{recall}, its parent node is guaranteed to be unable to satisfy Eq.~\eqref{eq:prune-c2}. Therefore, if there exists a hit set $\mathcal{S}^\prime$ that satisfies Eq. \eqref{eq:prune-c1} and Eq.~\eqref{eq:prune-c2}, then there is no node in $\mathcal{S}^\prime$ could be the parent node of any node in the hit set $\mathcal{S}^*$ generated by \ALGref{recall}. Furthermore, for each node $v^\prime$ in the hit set $\mathcal{S}^\prime$, we iteratively search for its parent node along its path to the root node. If there is no node on this path which belongs to $\mathcal{S}^*$, then the node $v^\prime$ has no contribution to covering the query range. The reason is that \ALGref{recall} ensures that the hit set $\mathcal{S}^*$ completely covers the range query. As a result, we can remove the node $v^\prime$ from the hit set $\mathcal{S}^\prime$ to make it smaller. 

So far, we have successfully proven two properties of an optimal hit set $\mathcal{S}_{opt}$: (1) for any node $v^*$ from the hit set $\mathcal{S}^*$ of \ALGref{recall}, no node from $\mathcal{S}_{opt}$ could exist on the path between the node $v^*$ and the root node; (2) for any node $v$ in an optimal hit set $\mathcal{S}_{opt}$, there must exist a node $v^* \in \mathcal{S}^*$ on the path between the root node and the node $v$. Therefore, the size of this optimal hit set is at least as large as the size of $\mathcal{S}^*$, which indicating $\mathcal{S} = \mathcal{S}_{opt}$.

\noindent\textbf{Logarithmic running time.} We want to further prove that the running time of \ALGref{recall} is $O(\log T)$. Since Algorithm 1 performs only $O(1)$ operations at each node, we only need to calculate the number of nodes traversed in the process of finding any range query in the segment tree. The proof here is quite similar with \ref{PRF:hit-height}. Similarly, we reconsider the three cases in \ref{PRF:hit-height}.
\begin{enumerate}
    \item The range query $\left[ T_s, T_e\right)$ is a prefix of the time range $\left[ \sigma_v, \tau_v \right)$, \ie, $T_s = \sigma_v$.
    \item The range query $\left[ T_s, T_e \right)$ is a suffix of the time range $\left[ \sigma_v, \tau_v \right)$, \ie, $T_e = \tau_v$.
    \item The range query $\left[ T_s, T_e \right)$ is neither a prefix nor a suffix of the time range $\left[ \sigma_v, \tau_v \right)$, \ie, $T_s \neq \sigma_v \text{ and } T_e \neq \tau_v$.
\end{enumerate}

Now, let $g_i(h)$ represent the number of traversed nodes in a stream segment tree with a height of $h$ in case ($i$), $v$ represent the root node, and $u_1$ and $u_2$ represent the left children and right children of $v$, respectively. Then, for case (1), if the range query can be entirely contained within the left sub-tree of $u_1$, then the traversed nodes will be the current root node $v$ and all the nodes traversed within the left sub-tree, \ie, $g_1(h)\leq 1 + g_1(h-1)$. Conversely, if the range query spans the time range of both the left sub-tree and right sub-tree, then the set of traversed nodes will be the current node $v$, the left children $u_1$, and all the nodes which will be traversed within the right sub-tree, \ie, $g_1(h) \leq 2+g_1(h-1)$. To sum up, we have the following expression
\begin{equation}
    g_1(h) \leq 2 + g_1(h-1).
\end{equation}
Therefore, considering the base case that $g_1(1) = 1$, we have 
\begin{equation}
    g_1(h) \leq 2 h -1.
\end{equation}
We can also have the same conclusion for case (2), \ie, $g_2(h) \leq 2 h -1$.

As for case (3), if the range query can be contained within the time range of any sub-tree, then similarly we count the current node $v$ and further recursively search within the sub-tree, \ie, $g_3(h) \leq 1 + g_3(h-1)$. Otherwise, please note that in this situation, the new query range will be either a prefix or suffix of the new sub-tree, which implies $g_3(h) \leq g_1(h -1) + g_2(h - 1)$. To sum up, we have 
\begin{equation}
    g_3(h) \leq 4 h -2.
\end{equation}

Therefore, given any stream segment tree with a height of $h$, the number of traversed nodes will be $O(h)$. Since $h = \lceil \log_2 T \rceil + 1$ in Theorem \ref{THM:alg-log-height}, we finally prove that the running time is $O(\log T)$.

\subsection{Proof of Lemma~\ref{LEM:time_complexity_temporal}}
Eq~\eqref{eq:stitch-temp-2} consists of three computations whose costs are as follows:
for $n = 2,...,p$ and $i = 1,...,s$,
(1) matrix multiplications $\BM U^{\Pa n\Tp}\TLD{\BM V}_i^\Pa n$, (2) tensor-matrix products between $\TLD{\BS H}_i$ and the preceding results $\BM U^{\Pa n\Tp}\TLD{\BM V}_i^\Pa n$, and 
(3) tensor-matrix products between the preceding results and matrices $\TLD{\BM V}_i^\Pa 1$ require  $O(spr^2D)$, $O(spr^{p+1})$, and $O(r^{p}(T_e-T_s))$, respectively.
Therefore, the complexity for computing Eq~\eqref{eq:stitch-temp-2} is $O(spr^2 D +r^{p}(T_e-T_s) + spr^{p+1})$.

\subsection{Proof of Lemma~\ref{LEM:time_complexity_nontemporal}}
In Eq~\eqref{eq:stitch-nontemp-3}, there are four computations:
for $m = 2,...,n-1,n+1,...,p$ and $i = 1,...,s$,
(1) matrix multiplications $\BM U_{[t_{i-1},t_i)}^{\Pa1\Tp}\TLD{\BM V}_i^\Pa1$,
(2) matrix multiplications $\BM U^{\Pa m\Tp}\TLD{\BM V}_i^\Pa m$, (3) tensor-matrix products between $\TLD{\BS H}_i$ and the preceding results $\BM U^{\Pa m\Tp}\TLD{\BM V}_i^\Pa m$, and 
(4) tensor-matrix products between the preceding results and matrices $\TLD{\BM V}_i^\Pa n$ require  $O(r^{2}(T_e-T_s))$, $O(spr^2D)$, $O(spr^{p+1})$, and $O(sr^{p}D)$, respectively.
Hence, the complexity for computing Eq~\eqref{eq:stitch-nontemp-3} is $O(spr^2 D + sr^{p} D + r^{2}(T_e-T_s) + spr^{p+1})$.

\subsection{Proof of Proposition~\ref{PROP:time_complexity_RQ}}
For each iteration, there are five computations: (1) the \textsc{Recall} algorithm, (2) the matricization of the temporal mode, (3) the matricization of $p-1$ non-temporal modes, (4) singular value decomposition $p$ times, and (5) the computation for updating core tensor.
Following the Lemmas~\ref{LEM:time_complexity_temporal} and~\ref{LEM:time_complexity_nontemporal}, the first two computations require $O(sp^2r^2 D + spr^{p} D + pr^{2}(T_e-T_s) + r^{p}(T_e-T_s) + sp^2r^{p+1})$.
Since we perform SVD for $p-1$ matrices of the size $D \times r^{p-1}$ and the matrix of the size $(T_e-T_s) \times r^{p-1}$, this requires $O(p\min(r^{p-1}D^2, r^{2(p-1)}D) + \min(r^{p-1}(T_e-T_s)^2,r^{2(p-1)} (T_e-T_s))$.
It can be expressed as $O(pr^{2(p-1)}D + r^{2(p-1)}(T_e-T_s))$ when $D > r^{p-1}$ and $(T_e-T_s) > r^{p-1}$. 
The last computation for a core tensor requires $O(r^{p}D)$.
Therefore, the time complexity is $O((sp^2r^2 + spr^{p} + pr^{2(p-1)})D + (pr^{2} + r^{p} + r^{2(p-1)})(T_e-T_s) + sp^2r^{p+1}+\log T)$ which is the sum of the complexities of these computations.

\subsection{Proof of Theorem~\ref{THM:alg-log-height}}
In this subsection, we will first prove that every insertion of \ALGref{append} is valid, and further prove that the height of the constructed stream segment tree is $\lceil \log_2 T \rceil + 1$ with $T$ representing the time span. 

First, to prove a valid insertion is equivalent to prove that, when a new node is inserted into a segment tree, the segment tree is not full. Please note that the design of creating nodes and insertion in \ALGref{append} ensures that, when a new node at the time of $T$ are inserted into the node $v$, the range query $\left[ \sigma_v, \tau_v \right)$ always contains the time of $T$, \ie, $\sigma \leq T < \tau_v$. Therefore, we only need to prove that for any node $v$ in a stream segment tree, the number of new nodes that can be inserted into the sub-tree rooted at $v$ is $\max \{\tau_v - T, 0 \}$. 

Specifically, we choose to mathematical induction to prove the above statement. When $T = 1$, this statement obviously holds since it means creating the first node for a stream segment tree. Now, let us assume the statement holds true when $T = t_0$. In this situation, if a node of $t_0$ is inserted within a sub-tree rooted at the node $v$, \ALGref{append} ensures that $\tau_v > t_0$, which implies $\max \{\tau_v - t_0, 0 \} > 0$. Then this insertion is valid since there is still space for a new node within the sub-tree of $v$. After this insertion, we can easily calculate that the number of new nodes that can still be inserted will be $\max \{ \tau_v -t_0, 0 \} - 1 = \max \{ \tau_v - (t_0 + 1), 0\}$. Therefore, the statement continues to hold when $T=t_0 + 1$, which completes the proof.

Second, we want to calculate the height of a constructed stream segment tree. From the previous discussion,  we can easily know that given a stream segment tree with a height of $h-1$, nodes can be continuously added until the stream segment tree is full. Once the tree becomes full, the attempt to add another node results in creating a new root node based on \ALGref{append}, with the original root node becoming the node of the left sub-tree. Consequently, the height of the new segment tree becomes $h$. Therefore, it is evident that the number of leaf nodes in a segment tree of height $h$ is greater than that of a perfect binary tree with height $h - 1$ and less than or equal to that of a perfect binary tree with height $h$. Then we successfully show that the height of a stream segment tree is $\lceil \log_2 T \rceil + 1$ since the number of leaf nodes is equal to $T$.

\subsection{Proof of Proposition~\ref{PROP:time_complexity_appending}}

We analyze the worst-case and amortized complexities of a stream segment tree.
Note that $s$ is equal to $2$.
We omit the number of iterations for the stitch operation and Tucker decomposition of the tensor slice for brevity.

First, performing Tucker-ALS of the tensor slice requires $O(prD^{p-1})$.
When we append a tensor slice $\BS X_T$, the worst-case number of the stitch operations is $O(\log T)$.
Then, updating the non-temporal factor matrices requires $O((p^2r^2 + pr^p + pr^{2(p-1)})D\log T + p^2r^{p+1}\log T)$ derived from the first and third terms of the time complexity in Proposition~\ref{PROP:time_complexity_RQ}.
Similarly, updating the temporal factor matrices requires $O((pr^2 + r^p + r^{2(p-1)})T)$ where the sum of the row sizes of the temporal factor matrices is $O(T)$.
Hence, the worst-case complexity of appending a tensor slice $\BS X_T$ is $O(prD^{p-1} + (p^2r^2 + pr^p + pr^{2(p-1)})D\log T + p^2r^{p+1}\log T + (pr^2 + r^p + r^{2(p-1)})T)$.

The amortized complexity is the time complexity of creating a stream segment tree divided by $T$.
We preprocess $T$ tensor slices of leaf nodes with the complexity $O(prD^{p-1}T)$.
For intermediate nodes, we perform $O(T)$ stitch operations with the complexity $O((p^2r^2 + pr^p + pr^{2(p-1)})DT+p^2r^{p+1}T + (pr^2 + r^p + r^{2(p-1)})T\log T)$.
Hence, the amortized complexity is $O(prD^{p-1} + (p^2r^2 + pr^p + pr^{2(p-1)})D+p^2r^{p+1} + (pr^2 + r^p + r^{2(p-1)})\log T)$.

\end{document}